\documentclass[preprint, 3p, twocolumn, times]{elsarticle}

\usepackage{times}
\usepackage{balance}
\usepackage{booktabs}
 \usepackage{siunitx}
\usepackage{url}

\usepackage{amsmath}
\usepackage{algorithm}
\usepackage{paralist}
\usepackage[noend]{algorithmic}
\usepackage{epsfig}
\usepackage{tikz}
\usepackage{subfig}
\graphicspath{{./figures/}{./figures/journalVersion/}}

\newcommand{\probf}{p_{\mathrm{false}}}
\newcommand{\setprobf}{\rho_{\mathrm{false}}}
\newcommand{\nbEl}{n}
\newcommand{\nbElExpected}{n_{\mathrm{exp}}}

\usetikzlibrary{chains,positioning}

\usepackage{todonotes}

\usepackage{graphicx}


\DeclareCaptionType{copyrightbox}

\newcommand\threefiguresnew[9]{%
    \def\tempa{#1}%
    \def\tempb{#2}%
    \def\tempc{#3}%
    \def\tempd{#4}%
    \def\tempe{#5}%
    \def\tempf{#6}%
    \def\tempg{#7}%
    \def\temph{#8}%
    \def\tempi{#9}%
    \threefigurescont
}
\newcommand\threefigurescont[1]{%
\begin{figure*}
\centering \subfloat[\tempb \label{\tempc}]{%
 \includegraphics[width=0.3 \textwidth]{\tempa}
 }
 \subfloat[\tempe \label{\tempf}]{%
 \includegraphics[width=0.3\textwidth]{\tempd}
 }
 \subfloat[\temph \label{\tempi}]{%
 \includegraphics[width=0.3\textwidth]{\tempg}
 }
 \caption{#1}
\end{figure*}
}

\newtheorem{theorem}{Theorem}
\newproof{proof}{Proof}

\journal{Information Systems}

\begin{document}
\begin{frontmatter}

\title{Bloofi: Multidimensional Bloom Filters}

\author[adina]{Adina Crainiceanu\corref{cor1}}
\ead{adina@usna.edu}
\address[adina]{US Naval Academy, USA}

\author[daniel]{Daniel Lemire}
\ead{lemire@gmail.com}
\address[daniel]{LICEF Research Center, TELUQ University of Quebec, Canada}

 \cortext[cor1]{Corresponding author. Tel.: 00+1+ 410 293-6822; fax: 00+1+410 293-2686.}


\begin{abstract}
Bloom filters are probabilistic data structures commonly used for approximate membership problems in many areas of Computer Science (networking, distributed systems, databa\-ses, etc.).
With the increase in data size and distribution of data, problems arise where a large number of Bloom filters are available, and all them need to be searched for potential matches.
As an example, in a federated cloud environment, each cloud provider could encode the information using Bloom filters and share the Bloom filters with a central coordinator. The problem of interest is not only whether a given element is in any of the sets represented by the Bloom filters, but which of the existing sets contain the given element. This problem cannot be solved by just constructing a Bloom filter on the union of all the sets. Instead, we effectively have a multidimensional Bloom filter problem: given an element, we wish to receive a list of candidate sets where the element might be.

To solve this problem, we consider 3~alternatives. Firstly, we can naively check many Bloom filters. Secondly, we propose to organize
the Bloom filters in a hierarchical index structure akin to a B+~tree, that we call  Bloofi. Finally, we propose another data structure that packs the Bloom filters in such a way as to exploit bit-level parallelism, which we call Flat-Bloofi.

Our theoretical and experimental results show that Bloofi and Flat-Bloofi provide  scalable and efficient solutions alternatives to search through a large number of Bloom filters.
\end{abstract}


\end{frontmatter}

\section{Introduction}

Bloom filters~\cite{Bloom:1970} are used to efficiently check whether an object is likely to be in the set (match) or whether the object is definitely not in the set (no match). False positives are possible, but false negatives are not. Due to their efficiency, compact representation, and flexibility in allowing a trade-off between space and false positive probability, Bloom filters are popular in representing diverse sets of data. They are used in databases~\cite{mullin90tose}, distributed systems~\cite{bigtable}, web caching~\cite{fan98webCaching}, and other network applications~\cite{Broder02networkapplications}.
For example, Google BigTable~\cite{Chang:2008:BDS:1365815.1365816} and Apache Cassandra~\cite{Vo:2010:TET:1920841.1920907} use Bloom filters to reduce the disk lookups for non-existent data.
 As digital data increases in both size and distribution, applications generate a large number of Bloom filters,
 and these filters need to be searched to find the sets containing particular objects.

Our work is
 motivated by  highly distributed data provenance applications, in which data is tracked as it is created, modified, or sent/received between the multiple sites participating in the application, each site maintaining the data in a cloud environment.
 Bloom filters can be maintained by each individual site and shared with a central location. For each piece of data, we need to find the sites holding the data. Thus, we may need to search through a large number of Bloom filters stored at the central location.

Indexing Bloom filters is different than indexing generic objects to improve search time. There is one level of indirection between the elements searched for, and the objects directly indexed by the index structure. In particular, each Bloom filter is a compact representation of an underlying set of elements. The question of interest is an \emph{all-membership} query: given a particular element (not a Bloom filter), which underlying sets contain that element? The query subject is an element, but the objects we are indexing and searching through are Bloom filters, so what we are creating is a meta-index. The traditional index structures, such as hash indexes, B+trees, R trees etc.\
and their distributed versions~\cite{Aguilera:2008:PSD:1453856.1453922} do not directly apply in this case as we are indexing Bloom filters and not the keys themselves. All we are given from each site is a Bloom filter.

There has been significant work in using Bloom filters in various applications, and developing variations of Bloom filters. Counting filters~\cite{fan98webCaching,4534126} support deletions from the Bloom filter; compressed Bloom filters~\cite{Mitzenmacher2001compressedBloomFIlter} are used with web caching; stable Bloom filters~\cite{Deng2006stableBloomFilters} eliminate duplicates in streams, spectral Bloom filters~\cite{Cohen2003spectralBloomFilter} extend the applicability of Bloom filters to multi-sets, multi-class Bloom Filter (MBF)~\cite{Li:2011:SDC:2117689.2118897} use per-element probabilities. Yet there has been few attempts to accelerate queries over   many Bloom filters, what we call the \emph{multidimensional Bloom filter problem}, even though our problem is  closely related to signature file methods (see Section~\ref{sec:relatedWork}) where one seeks to index set-value attributes.

To solve this problem, we propose \emph{Bloofi} (\textbf{Bloo}m \textbf{F}ilter \textbf{I}ndex), a hierarchical index structure for Bloom filters. Bloofi provides probabilistic answers to \emph{all-membership} queries and scales to tens of thousands of Bloom filters.
When the probability of false positives is low, Bloofi of order $d$ (a tunable parameter) can provide $O(d \log_d N)$ search cost, where $N$ is the number of Bloom filters indexed. Bloofi also provides support for inserts, deletes, and updates with $O(d \log_d N)$ cost and requires $O(N)$ storage cost.
In designing Bloofi, we take advantage of the fact that the bitwise OR between Bloom filters of same length, constructed using the same hash functions, is also a Bloom filter. The resulting Bloom filter represents the union of the sets represented by the individual Bloom filters. This property allows us to construct a tree where the leaf levels are the indexed Bloom filters, and the root level is a Bloom filter that represents all the elements in the system. This tree is used to prune the search space (eliminate Bloom filters as candidates for matches) while processing \emph{all-membership} queries. Our performance evaluation shows that Bloofi performs best when the false positive probability of the union Bloom filter (a Bloom filter that is the union of all the indexed Bloom filters) is low and provides $O(d \times \log_d N)$ search performance in most cases, with $O(N)$ storage cost and $O(d \times \log_d N)$ maintenance cost.
 Bloofi could be used whenever a large number of Bloom filters that use the same hash functions need to be checked for matches.

Bloom filters are constructed over bitmaps, i.e., vector of Booleans. With bitmaps, we can exploit bit-level parallelism: on a 64-bit processor, we can compute the bitwise or between
64~bits using a single instruction. We use
bit-level parallelism with Bloofi to optimize the construction of the data structure.
However, we have also designed an alternative
data structure that is designed specifically to exploit bit-level parallelism (henceforth Flat-Bloofi). Though not as scalable as Bloofi, it can be fast when the number of Bloom filters is moderate.

This article is an extended version of ``Bloofi: A Hierarchical Bloom Filter Index with Applications to Distributed Data Provenance''~\cite{crainiceanu2013bloofi} published in the 2nd International Workshop on Cloud Intelligence Cloud-I 2013. The paper was completely revised, and the new version introduces an additional data structure, Flat-Bloofi, a new implementation for Bloofi that improves the performance by an order of magnitude, and a new performance evaluation.

The rest of this paper is structured as follows:
Section~\ref{sec:distributedProvenance} describes a distributed data provenance application for Bloofi.
Section~\ref{sec:bloom} briefly reviews the concept of Bloom filter.
 Section~\ref{sec:indexing} introduces Bloofi, a hierarchical index structure for Bloom filters. Section~\ref{sec:maintenance} introduces the maintenance algorithms and a theoretical performance analysis.
  Section~\ref{sec:flat} introduces Flat-Bloofi, a data structure for the multidimensional Bloom filter problem, designed to exploit bit-level parallelism.
 Section~\ref{sec:experiments} shows the experimental results. We discuss related work in Section~\ref{sec:relatedWork} and conclude in Section~\ref{sec:conclusions}.

\section{Motivation: application to distributed data provenance}\label{sec:distributedProvenance}
In this section we describe the distributed data provenance application that motivated our work on Bloofi.
Let us assume that a multinational corporation with hundreds of offices in geographically distributed locations (sites) around the world is interested in tracking the documents produced and used within the corporation. Each document is given a universally unique identifier (uuid) and is stored in the local repository, in a cloud environment. Documents can be sent to another location (site) or received from other locations, multiple documents can be bundled together to create new documents, which therefore are identified by new uuids, documents can be decomposed in smaller parts that become documents themselves, and so on. All these ``events'' that are important to the provenance of a document are recorded in the repository at the site generating the event. The events can be stored as RDF triples in a scalable cloud triple store such as Rya~\cite{rya}. The data can be modeled as a Directed Acyclic Graph (DAG), with labeled edges (event names) and nodes (document uuids). As documents travel between sites, the DAG is in fact distributed not only over the machines in the cloud environment at each site, but also over hundreds of geographically distributed locations. The data provenance problem we are interested in solving is finding all the events and document uuids that form the ``provenance'' path of a given uuid (all ``ancestors'' of a given node in the distributed graph).

Storing all the data, or even all the uuids and their location, in a centralized place is not feasible, due to the volume and speed of the documents generated globally. Fully distributed data structures, such as Chord~\cite{stoica2001chord} or P-Ring~\cite{crainiceanu2007pring}, require even more communication (messages) than the centralized solution, increasing the latency and bandwidth consumption, so they are also not feasible due to the volume and speed of the documents generated globally. Moreover, local regulations might impose restrictions on where the data can be stored. However, since all the global locations belong to the same corporation, data exchange and data tracking must be made possible.

Without any information on the location of a uuid, each provenance query for a uuid must be sent to all sites. Each site can then determine the local part of the provenance path, and return it. However, the provenance might contain new uuids, so a new query needs to be sent to each site for each new uuid connected to the original uuid, until no new uuids are found. This recursive process could consume significant bandwidth and latency in a geographically distributed system.

To minimize the number of unnecessary messages sent to determine the full provenance of an object, each local site maintains a Bloom filter of all the uuids in the local system. Updates to the Bloom filter are periodically propagated to a centralized location (the headquarters for example). Since the Bloom filters are compact representations of the underlying data, less bandwidth is consumed when Bloom filters are sent versus sending the actual data. At the central location, a Bloofi index is constructed from the Bloom filters, and every time a provenance query for a uuid is made, the Bloofi index is used to quickly determine the sites that might store provenance information for the given uuid. If the query load is too high and having the Bloofi index in a single location affects performance and/or availability, multiple Bloofi indexes could be constructed in several locations.

\section{Bloom filters}
\label{sec:bloom}

Checking for the presence of a value in a set or similar  data structure can be expensive, especially if the
data structure is stored on disk. When we expect
many requests for values that are not present, it is helpful to have an auxiliary data structure that will quickly dismiss these requests.   Bloom filters can serve this purpose.

Each Bloom filter~\cite{Bloom:1970} is a bit array (or bitmap) of length $m$ constructed by using a set of $k$~hash functions. The empty Bloom filter has all bits 0. To add an element to the filter, each of the $k$~hash functions maps the new element to a position in the bit array. The bit in that position is turned to 1. To check whether an element is a member of the set represented by the Bloom filter, the $k$~hash functions are applied to the test element. If any of the resulting $k$ positions is 0, the test element is not in the set, with probability 1. If all $k$ positions are 1, the Bloom filter \emph{matches} the test element, and the test element might be in the set (it might be a true positive or a false positive). There is a trade-off between the size of the Bloom filter and the probability of false positives, $\probf{}$, returned by it. We have that $\probf{} \approx ( 1-e^{-kn/m} )^k$ assuming that there are $n$~elements in the filter~\cite{Mitzenmacher:2005:PCR:1076315}; the probability is minimal
when $k= m/n \ln 2$. The probability $\probf{}$ can be lowered by increasing the size of the Bloom filter ($m$). For a fixed number of elements ($n$), the probability goes to zero exponentially. In the rest of the paper we assume that all the Bloom filters indexed have the same length and use the same set of hash functions.

\section{Indexing Bloom filters}\label{sec:indexing}

We are given a collection of $N$~Bloom filters $\mathcal{B}$. Given a query for a value, we want to find all Bloom filters that are a match.
We are most interested in the case where there are relatively few such Bloom filters and when a sequential search is potentially inefficient.

\subsection{Bloofi: a hierarchical Bloom filter index}

Bloofi, the Bloom filter index, is based on the following idea. We construct a tree: the leaves of the tree are the Bloom filters to be indexed, and the parent nodes are Bloom filters obtained by applying a bitwise OR on the child nodes. This process continues until the root is reached. The index has the property that each non-leaf Bloom filter in the tree represents the union of the sets represented by the Bloom filters in the sub-tree rooted at that node. As a consequence, if an object matches a leaf-level Bloom filter, it matches all the Bloom filters in the path from that leaf to the root. Conversely, if a particular Bloom filter in Bloofi does not match an object, there is no match in the entire sub-tree rooted at that node.

Using Bloofi, a membership query starts by first querying the root Bloom filter: if it does not match the queried object, then none of the indexed sets contain the object and a negative answer is returned. If the root does match the object, the query proceeds by checking which of the child Bloom filters matches the object. The query continues down the path of Bloom filters matching the object, until the leaf level is reached. In a balanced tree, the height of Bloofi is logarithmic in the number of Bloom filters indexed, and each step in the query process goes down one level in the tree. In the best case, a query with a negative answer is answered in constant time (check the root only), and a query with a positive answer is answered in logarithmic time. However, if multiple paths in the index are followed during the query process, the query time increases. Section~\ref{sec:insert} introduces our heuristics for Bloofi construction such that the number of ``misleading'' paths in the Bloofi is reduced (similar Bloom filters are in the same sub-tree).

There are many possible implementations for Bloofi: as a tree similar with binary search trees, AVL trees, B+~trees, etc. Due to the flexibility allowed by having a balanced tree with the number of child pointers higher than two, we implement Bloofi like a B+~trees.
We start with an \emph{order} parameter $d$, and each non-leaf node maintains $l$~child pointers: $d \le l \le 2d$ for all non-root nodes, and $2 \le l \le 2d$ for the root.
Each node in Bloofi stores only one value, which is different than general search trees. For the leaves, the value is the Bloom filter to be indexed. For all non-leaf nodes, the value is obtained by applying bitwise OR on the values of its child nodes. Throughout the paper we use the usual definitions for tree, node in a tree, root, leaf, depth of a node (the number of edges from root to the node), height of a tree (maximum depth of a node in the tree), sibling, and parent.

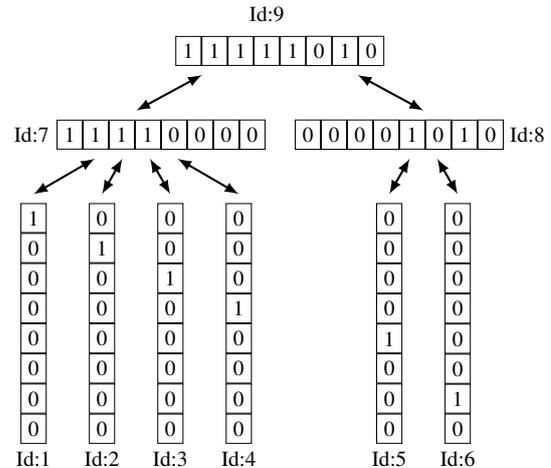
\begin{figure}[htbp]
\centering
\begin{tikzpicture}[edge from parent/.style={draw,latex-latex},thick,scale=0.9, every node/.style={scale=0.9},
start chain=1 going below,
start chain=2 going below,
start chain=3 going below,
start chain=4 going below,
start chain=5 going below,
start chain=6 going below,
start chain=7 going right,
start chain=8 going right,
start chain=9 going right,
child anchor=north,
level distance=1.5cm,
level 1/.style={sibling distance=42mm},
level 2/.style={sibling distance=10mm,level distance=3cm}]
\edef\sizetape{0.2cm}
\node{%
\begin{tikzpicture}[node distance=0cm]
    \tikzstyle{mytape}=[draw,minimum size=\sizetape]
    \node  [on chain=9,mytape] {1};
    \node [on chain=9,mytape] {1};
    \node [on chain=9,mytape] {1};
    \node(id9) [on chain=9,mytape] {1};
    \node [on chain=9,mytape] {1};
        \node [above of=id9,node distance=0.6cm] {Id:9};
    \node [on chain=9,mytape] {0};
    \node [on chain=9,mytape] {1};
    \node [on chain=9,mytape] {0};
\end{tikzpicture}
}
child{node {\begin{tikzpicture}[node distance=0cm]
    \tikzstyle{mytape}=[draw,minimum size=\sizetape]
    \node (id7) [on chain=7,mytape] {1};
      \node [left =of id7,node distance=0.6cm] {Id:7};
    \node [on chain=7,mytape] {1};
    \node [on chain=7,mytape] {1};
    \node [on chain=7,mytape] {1} ;
    \node [on chain=7,mytape] {0};
    \node [on chain=7,mytape] {0};
    \node [on chain=7,mytape] {0};
    \node [on chain=7,mytape] {0};
\end{tikzpicture}
}
child {node{\begin{tikzpicture}[node distance=0cm]
    \tikzstyle{mytape}=[draw,minimum size=\sizetape]
    \node [on chain=1,mytape] {1};
    \node [on chain=1,mytape] {0};
    \node [on chain=1,mytape] {0};
    \node [on chain=1,mytape] {0} ;
    \node [on chain=1,mytape] {0};
    \node [on chain=1,mytape] {0};
    \node [on chain=1,mytape] {0};
    \node (id1)  [on chain=1,mytape] {0};
      \node [below =of id1,node distance=0.7cm] {Id:1};
\end{tikzpicture}}} child {node{\begin{tikzpicture}[node distance=0cm]
    \tikzstyle{mytape}=[draw,minimum size=\sizetape]
    \node [on chain=2,mytape] {0};
         \node [on chain=2,mytape] {1};
    \node [on chain=2,mytape] {0};
    \node [on chain=2,mytape] {0} ;
    \node [on chain=2,mytape] {0};
    \node [on chain=2,mytape] {0};
    \node [on chain=2,mytape] {0};
    \node (id2) [on chain=2,mytape] {0};
    \node [below =of id2,node distance=0.7cm] {Id:2};
\end{tikzpicture}}} child {node{\begin{tikzpicture}[node distance=0cm]
    \tikzstyle{mytape}=[draw,minimum size=\sizetape]
    \node [on chain=3,mytape] {0};
         \node [on chain=3,mytape] {0};
    \node [on chain=3,mytape] {1};
    \node [on chain=3,mytape] {0} ;
    \node [on chain=3,mytape] {0};
    \node [on chain=3,mytape] {0};
    \node [on chain=3,mytape] {0};
    \node (id3) [on chain=3,mytape] {0};
    \node [below =of id3,node distance=0.7cm] {Id:3};
\end{tikzpicture}}} child {node{\begin{tikzpicture}[node distance=0cm]
    \tikzstyle{mytape}=[draw,minimum size=\sizetape]
    \node [on chain=4,mytape] {0};
         \node [on chain=4,mytape] {0};
    \node [on chain=4,mytape] {0};
    \node [on chain=4,mytape] {1} ;
    \node [on chain=4,mytape] {0};
    \node [on chain=4,mytape] {0};
    \node [on chain=4,mytape] {0};
    \node (id4) [on chain=4,mytape] {0};
    \node [below =of id4,node distance=0.6cm] {Id:4};
\end{tikzpicture}}} }
child{node {\begin{tikzpicture}[node distance=0cm]
    \tikzstyle{mytape}=[draw,minimum size=\sizetape]
    \node [on chain=8,mytape] {0};
    \node [on chain=8,mytape] {0};
    \node [on chain=8,mytape] {0};
    \node [on chain=8,mytape] {0} ;
    \node [on chain=8,mytape] {1};
    \node [on chain=8,mytape] {0};
    \node [on chain=8,mytape] {1};
    \node (id8) [on chain=8,mytape] {0};
     \node  [right of= id8,node distance=0.6cm] {Id:8};
\end{tikzpicture}} child {node{\begin{tikzpicture}[node distance=0cm]
    \tikzstyle{mytape}=[draw,minimum size=\sizetape]
    \node [on chain=5,mytape] {0};
         \node [on chain=5,mytape] {0};
    \node [on chain=5,mytape] {0};
    \node [on chain=5,mytape] {0} ;
    \node [on chain=5,mytape] {1};
    \node [on chain=5,mytape] {0};
    \node [on chain=5,mytape] {0};
    \node (id5) [on chain=5,mytape] {0};
    \node [below =of id5,node distance=0.7cm] {Id:5};
\end{tikzpicture}}}  child {node{\begin{tikzpicture}[node distance=0cm]
    \tikzstyle{mytape}=[draw,minimum size=\sizetape]
    \node [on chain=6,mytape] {0};
         \node [on chain=6,mytape] {0};
    \node [on chain=6,mytape] {0};
    \node [on chain=6,mytape] {0} ;
    \node [on chain=6,mytape] {0};
    \node [on chain=6,mytape] {0};
    \node [on chain=6,mytape] {1};
    \node (id6) [on chain=6,mytape] {0};
    \node [below =of id6,node distance=0.7cm] {Id:6};
\end{tikzpicture}}} };
\end{tikzpicture}
\caption{Bloofi Tree of Order 2\label{fig:bloofiTree}}

\label{fig:bloofiTreeÉ}
\end{figure}

Fig.~\ref{fig:bloofiTree} shows an example of a Bloofi index of order 2. Each internal node has between 2 and 4 child pointers. The leaf level of the tree contains the original Bloom filters indexed by the Bloofi index, with identifiers 1, 2, 3, 4, 5, and 6. The node identifiers are shown here for ease of presentation, but they are also used in practice during updates and deletes to identify the node that needs to be updated or deleted (see Section~\ref{sec:delete} and Section~\ref{sec:update} for details). In the rest of the paper, we often use ``node X'' to refer to the node with identifier X. At the next higher level, the values are obtained by applying bitwise OR on the values of the children nodes, so the value of the node 7 is the bitwise OR between values of nodes 1, 2, 3, and 4. The process continues until the root is reached.

We use the following notation.
Given a node, its parent is written node.parent. The ordered set of its children is written by node.children. The number of children is written node.nbDesc with the convention that node.nbDesc is zero for a leaf node.
For $i=0, \ldots \text{node.nbDesc}-1$, $\text{node.children}[i]$ is the $(i+1)^{\mathrm{th}}$ children of the node.

Each node has a corresponding bit array indicated
by node.val. We let $|$ be the bitwise OR operation so that when node is not a leaf, the bit array is just the aggregation of the bit array of the children: \begin{align*} \text{node.val}  = &   \text{node.children}[0].\text{val} \, \big | \,\text{node.children}[1].\text{val} \, \big | \, \cdots  & \\
& \big |  \,\text{node.children}[\text{node.nbDesc}-1].\text{val}. &\end{align*}
Individual bit values in the bit array are accessed as $\text{node.val}[i]$ for $i =0,\ldots, m$.

Since Bloofi is a balanced tree, with each internal node having at least $d$ child nodes, where $d$ is the order of the tree, the height of a Bloofi index of order $d$ is at most $\lfloor log_{d} N\rfloor$, where $N$ is the number of Bloom filters to index.

\subsection{Search}\label{sec:search}
The search algorithm (Algorithm~\ref{alg:findMatches}) returns the identifiers of all the leaf-level Bloom filters that match a given object in the subtree rooted at the given node. It first checks whether the current node value matches the object (line~\ref{line:match1}).
If not, then none of the Bloom filters in that sub-tree match the object, so the empty set is returned (line~\ref{line:ermpty}). If the current node does match the object, then either it is a leaf, in which case it returns the identifier (line~\ref{line:reidfir}), or it is an inner node, in which case the \texttt{findMatches} function is called recursively for all of its child nodes (lines~\ref{line:beginrecur}--\ref{line:endrecur}).

\begin{algorithm}[htbp]
\caption{: \texttt{findMatches}(node,$o$)}\label{alg:findMatches}
\begin{algorithmic}[1]\small
\STATE //RETURN VALUE: the identifiers of leaves in the subtree rooted at node with Bloom filters matching the object~$o$
\STATE //if node does not matches the object, return empty set, else check the descendants
\IF{not match(node.val,$o$)}\label{line:match1}
\RETURN $\emptyset$;\label{line:ermpty}
\ELSE
\STATE //if this node is a leaf, just return the identifier
\IF{$\text{node.nbDesc} = 0$ }
\RETURN $node.id$;\label{line:reidfir}
\ELSE
\STATE //if not leaf, check the descendants
\STATE returnList = $\emptyset$;
\FOR{\label{line:beginrecur}$i = 0$; $i < \text{node.nbDesc}$; i++}
\STATE returnList.add(findMatches(node.children[$i$],$o$));
\ENDFOR\label{line:endrecur}
\ENDIF
\ENDIF
\RETURN returnList;
\end{algorithmic}
\end{algorithm}

\textbf{Example.} Consider a query for object value $4$ in the Bloofi tree in Figure~\ref{fig:bloofiTree}. The \texttt{findMatches} function in Algorithm~\ref{alg:findMatches} is invoked with arguments root and $4$. In line 3 of Algorithm~\ref{alg:findMatches}, the algorithm checks whether the value of the root matches $4$. For simplicity of presentation, assume that there is only one hash function used by the Bloom filters, the function is $h(x) = x \bmod{8}$, and the elements in the underlying set are integers. Since $\text{root.val}[4]$ is 1, the root matches the queried object $4$ and the search proceeds by invoking the \texttt{findMatches} function for each of its child nodes. The first child node, node 7, does not match the queried object, so \texttt{findMatches} for that sub-tree returns $\emptyset$. The second child node of the root, node 8, matches $4$, so the search continues at the next lower level. Node 5 matches the queried object, and is a leaf, so the \texttt{findMatches} function returns the identifier 5. Leaf 6 does not match the queried object, so that \texttt{findMatches} call returns $\emptyset$. Now, the recursive call on node 8 returns with the value \{5\}, and finally the call on the root returns \{5\}, which is the result of the query.

\paragraph{Search cost} The complexity of the search process is given by the number of \texttt{findMatches} invocations. In the best case, if there are no leaf-level Bloom filters matching a given object, the number of Bloom filters to be checked for matches is 1 (the root). To find a leaf-level Bloom filter that matches a query, the number of \texttt{findMatches} invocations is $O(d \log_d N)$ in the best case (one path is followed, and at each node, all children are checked to find the one that matches) and $O(N)$ in the worst case (since the maximum number of nodes in a Bloofi tree is $\lceil N + (N-1)/(d-1)\rceil$, the search cost is $O(N)$ if all nodes need to be checked for matches).

\section{Bloofi maintenance}\label{sec:maintenance}
We introduce the algorithms for inserting, deleting, and updating Bloom filters.

\subsection{Insert}\label{sec:insert}

Ideally, when inserting  Bloom filters, we would like to keep them in  partitions
so that the overlap between different partitions is small. That is, we would
like \emph{similar} Bloom filters to be grouped together as much as possible.
And, conversely, we would like Bloom filters from different partitions to be as different as possible. Such problems are commonly NP-hard (e.g., the Minimum Graph Bisection Problem) though they can be sometimes approximated
efficiently. We leave a more formal investigation of this problem to future work and use a heuristic.

The algorithm for inserting (Algorithm~\ref{alg:insert}) finds a leaf which is ``close'' to the input Bloom filter in a given metric space, and inserts the new Bloom filter next to that leaf. The intuition is that similar Bloom filters should be in the same sub-tree to improve search performance. As distance metric, we use the Hamming distance. That is, we count the number of bits that differ. This can be computed quickly by computing the cardinality of the bitwise exclusive OR of  two bit arrays along with fast functions to count the number of 1s in the resulting words (e.g., \texttt{Long.bitCount} in Java). We could consider other distance metric, and we experiment with Cosine and Jaccard metrics in Section~\ref{sec:varyMetric}.

 The new Bloom filter is inserted by first updating the value of the current node by computing the bitwise or with the value of the filter to be inserted (since that node will be in the sub-tree), and then recursively calling the insert function on the child node most similar with the new value (line~\ref{line:recursive}).
 Once the most similar leaf node is located, a new leaf is created for the new Bloom filter (line~\ref{line:newcreated}) and is inserted as a sibling of the node by calling the \texttt{insertIntoParent} function (Algorithm~\ref{alg:insertEntry}). This function takes as parameters the new node newEntry, and the most similar node to it, node. We insert newEntry  as a sibling of node. If the number of children in the parent is still at most $2d$, the insert is complete. If an overflow occurs, the node splits (lines 9-16) and the newly created node is returned by the function. The splits could occasionally propagate up to the root level. In that case, a new root is created, and the height of the Bloofi tree increases (line~\ref{line:newroot}).

\begin{algorithm}[htbp]
\caption{: \texttt{insert}(newBloomFilter, node)}\label{alg:insert}
\begin{algorithmic}[1]\small
\STATE // $D$ is a distance function between bit arrays
\STATE //insert into the sub-tree rooted at the given node
\STATE //RETURN: null or pointer to new child if split occurred
\STATE //if node is not leaf, direct the search for the new filter place
\IF{$\text{node.nbDesc} > 0$}
\STATE //update the value of the node to contain the new filter
\STATE $\text{node.val} = \text{node.val}\, \big | \,\text{newBloomFilter}$;
\STATE //find the most similar child and insert there
\STATE find child $C$  minimizing
$D(C.\text{val}, \text{newBloomFilter})$
\STATE \label{line:recursive} $\text{newSibling} = \text{\texttt{insert}}(\text{newBloomFilter}, C)$;
\STATE //if there was no split, just return null
\IF{$\text{newSibling} = \text{null}$}
\RETURN null;
\ELSE
\STATE //there was a split; check whether a new root is needed
\IF{$\text{node.parent} = \text{null}$}
\STATE //root was split; create a new root
\STATE \label{line:newroot} $\text{newRoot} = \text{new BFINode()}$ ; // create new node
\STATE $\text{newRoot.val} = \text{node.val}\, \big | \, \text{newSibling.val}$;
\STATE $\text{newRoot.parent} = \text{null}$;
\STATE $\text{newRoot.children.add}(\text{node})$;
\STATE $\text{newRoot.children.add}(\text{newSibling})$;
\STATE $\text{root} = \text{newRoot}$;
\STATE $\text{node.parent} = \text{newRoot}$;
\STATE $\text{newSibling.parent} = \text{newRoot}$;
\RETURN null;
\ELSE
\STATE $\text{newSibling} = \texttt{insertIntoParent}(\text{newSibling}, \text{node})$
\RETURN newSibling;
\ENDIF //current node is root or not
\ENDIF //there was a split or not
\ELSE
\STATE //if node is leaf, need to insert into the parent
\STATE //create a node for newBloomFilter
\STATE \label{line:newcreated} $\text{newLeaf} = \text{new BFINode()}$;
\STATE $\text{newLeaf.val} = \text{newBloomFilter}$;
\STATE //insert the new leaf into the parent node
\STATE $\text{newSibling} = \texttt{insertIntoParent}(\text{newLeaf},\text{node})$;
\RETURN newSibling;
\ENDIF //current node is leaf or not
\end{algorithmic}
\end{algorithm}

\begin{algorithm}[thbp]
\caption{: \texttt{insertIntoParent}(newEntry, node)}\label{alg:insertEntry}
\begin{algorithmic}[1]\small
\STATE //insert into the node's parent, after the node pointer
\STATE //RETURN: null or pointer to new child if split occurred
\STATE $\text{node.parent.children.addAfter}(\text{newEntry}, \text{node})$;
\STATE $\text{newEntry.parent} = \text{node.parent}$;
\STATE //check for overflow
\IF{$\text{node.nbDesc} > 2d$ }
\RETURN null;
\ELSE
\STATE //overflow, so split
\STATE $P = \text{node.parent}$;
\STATE $P' = \text{new BFINode()}$;
\STATE move last $d$ children from $P$ to $P'$;
\STATE update parent information for all children of $P'$;
\STATE re-compute $P.\text{val}$ as the OR between its children values;
\STATE compute $P'.\text{val}$ as the OR between its children values;
\RETURN $P'$;
\ENDIF
\end{algorithmic}
\end{algorithm}

\begin{figure}[htbp]
\resizebox{\columnwidth}{!}{\begin{tikzpicture}[edge from parent/.style={draw,latex-latex},
scale=0.9, every node/.style={scale=0.9},
start chain=1 going below,
start chain=2 going below,
start chain=3 going below,
start chain=4 going below,
start chain=5 going below,
start chain=6 going below,
start chain=10 going below,
start chain=7 going right,
start chain=8 going right,
start chain=9 going right,
start chain=11 going right,
child anchor=north,
level distance=1.5cm,
level 1/.style={sibling distance=34mm},
level 2/.style={sibling distance=10mm,level distance=3.5cm}]
\edef\sizetape{0.2cm}
\node{%
\begin{tikzpicture}[node distance=0cm]
    \tikzstyle{mytape}=[draw,minimum size=\sizetape]
    \node(id9)  [on chain=9,mytape] {1};
    \node [on chain=9,mytape] {1};
    \node [on chain=9,mytape] {1};
    \node [on chain=9,mytape] {1};
    \node [on chain=9,mytape] {1};
    \node [on chain=9,mytape, blue, font=\itshape] {1};
    \node [on chain=9,mytape] {1};
    \node [on chain=9,mytape] {0};
    \node [above of=id9,node distance=0.4cm] {Id:9};
\end{tikzpicture}
}
child{node {\begin{tikzpicture}[node distance=0cm]
    \tikzstyle{mytape}=[draw,minimum size=\sizetape]
    \node (id7) [on chain=7,mytape,blue,font=\itshape] {1};
      \node [above of =id7,node distance=0.4cm] {Id:7};
    \node [on chain=7,mytape,blue,font=\itshape] {1};
    \node [on chain=7,mytape,blue,font=\itshape] {1};
    \node [on chain=7,mytape,blue,font=\itshape] {0} ;
    \node [on chain=7,mytape,blue,font=\itshape] {0};
    \node [on chain=7,mytape,blue,font=\itshape] {0};
    \node [on chain=7,mytape,blue,font=\itshape] {0};
    \node [on chain=7,mytape,blue,font=\itshape] {0};
\end{tikzpicture}
}
child {node{\begin{tikzpicture}[node distance=0cm]
    \tikzstyle{mytape}=[draw,minimum size=\sizetape]
    \node [on chain=1,mytape] {1};
    \node [on chain=1,mytape] {0};
    \node [on chain=1,mytape] {0};
    \node [on chain=1,mytape] {0} ;
    \node [on chain=1,mytape] {0};
    \node [on chain=1,mytape] {0};
    \node [on chain=1,mytape] {0};
    \node (id1)  [on chain=1,mytape] {0};
      \node [below =of id1,node distance=0.7cm] {Id:1};
\end{tikzpicture}}} child {node{\begin{tikzpicture}[node distance=0cm]
    \tikzstyle{mytape}=[draw,minimum size=\sizetape]
    \node [on chain=2,mytape] {0};
         \node [on chain=2,mytape] {1};
    \node [on chain=2,mytape] {0};
    \node [on chain=2,mytape] {0} ;
    \node [on chain=2,mytape] {0};
    \node [on chain=2,mytape] {0};
    \node [on chain=2,mytape] {0};
    \node (id2) [on chain=2,mytape] {0};
    \node [below =of id2,node distance=0.7cm] {Id:2};
\end{tikzpicture}}} child {node{\begin{tikzpicture}[node distance=0cm]
    \tikzstyle{mytape}=[draw,minimum size=\sizetape]
    \node [on chain=3,mytape] {0};
         \node [on chain=3,mytape] {0};
    \node [on chain=3,mytape] {1};
    \node [on chain=3,mytape] {0} ;
    \node [on chain=3,mytape] {0};
    \node [on chain=3,mytape] {0};
    \node [on chain=3,mytape] {0};
    \node (id3) [on chain=3,mytape] {0};
    \node [below =of id3,node distance=0.7cm] {Id:3};
\end{tikzpicture}}} }
child{node {\begin{tikzpicture}[node distance=0cm]
    \tikzstyle{mytape}=[draw,minimum size=\sizetape]
    \node [on chain=11,mytape,blue, font=\itshape] {0};
    \node (id11) [on chain=11,mytape,blue,font=\itshape] {0};
    \node [on chain=11,mytape,blue,font=\itshape] {1};
    \node [on chain=11,mytape,blue,font=\itshape] {1} ;
    \node [on chain=11,mytape,blue,font=\itshape] {0};
    \node [on chain=11,mytape,blue,font=\itshape] {1};
    \node [on chain=11,mytape,blue,font=\itshape] {0};
    \node [on chain=11,mytape,blue,font=\itshape] {0};
    \node [above of= id11,node distance=0.4cm,blue,font=\itshape] {Id:11};
\end{tikzpicture}} child {node{\begin{tikzpicture}[node distance=0cm]
    \tikzstyle{mytape}=[draw,minimum size=\sizetape]
    \node [on chain=10,mytape,blue,font=\itshape] {0};
    \node [on chain=10,mytape,blue,font=\itshape] {0};
    \node [on chain=10,mytape,blue,font=\itshape] {1};
    \node [on chain=10,mytape,blue,font=\itshape] {0} ;
    \node [on chain=10,mytape,blue,font=\itshape] {0};
    \node [on chain=10,mytape,blue,font=\itshape] {1};
    \node [on chain=10,mytape,blue,font=\itshape] {0};
    \node (id10) [on chain=10,mytape,blue,font=\itshape] {0};
    \node [below =of id10,node distance=0.7cm,blue,font=\itshape] {Id:10};
\end{tikzpicture}}}  child {node{\begin{tikzpicture}[node distance=0cm]
    \tikzstyle{mytape}=[draw,minimum size=\sizetape]
    \node [on chain=4,mytape] {0};
         \node [on chain=4,mytape] {0};
    \node [on chain=4,mytape] {0};
    \node [on chain=4,mytape] {1} ;
    \node [on chain=4,mytape] {0};
    \node [on chain=4,mytape] {0};
    \node [on chain=4,mytape] {0};
    \node (id4) [on chain=4,mytape] {0};
    \node [below =of id4,node distance=0.7cm] {Id:4};
\end{tikzpicture}}} }
child{node {\begin{tikzpicture}[node distance=0cm]
    \tikzstyle{mytape}=[draw,minimum size=\sizetape]
    \node (id8)[on chain=8,mytape] {0};
    \node [on chain=8,mytape] {0};
    \node [on chain=8,mytape] {0};
    \node [on chain=8,mytape] {0} ;
    \node [on chain=8,mytape] {1};
    \node [on chain=8,mytape] {0};
    \node [on chain=8,mytape] {1};
    \node [on chain=8,mytape] {0};
     \node  [above of= id8,node distance=0.4cm] {Id:8};
\end{tikzpicture}} child {node{\begin{tikzpicture}[node distance=0cm]
    \tikzstyle{mytape}=[draw,minimum size=\sizetape]
    \node [on chain=5,mytape] {0};
         \node [on chain=5,mytape] {0};
    \node [on chain=5,mytape] {0};
    \node [on chain=5,mytape] {0} ;
    \node [on chain=5,mytape] {1};
    \node [on chain=5,mytape] {0};
    \node [on chain=5,mytape] {0};
    \node (id5) [on chain=5,mytape] {0};
    \node [below =of id5,node distance=0.7cm] {Id:5};
\end{tikzpicture}}}  child {node{\begin{tikzpicture}[node distance=0cm]
    \tikzstyle{mytape}=[draw,minimum size=\sizetape]
    \node [on chain=6,mytape] {0};
         \node [on chain=6,mytape] {0};
    \node [on chain=6,mytape] {0};
    \node [on chain=6,mytape] {0} ;
    \node [on chain=6,mytape] {0};
    \node [on chain=6,mytape] {0};
    \node [on chain=6,mytape] {1};
    \node (id6) [on chain=6,mytape] {0};
    \node [below =of id6,node distance=0.7cm] {Id:6};
\end{tikzpicture}}} };
\end{tikzpicture}}
\caption{Bloofi Tree After Insert and Split}
\label{fig:bloofiTree_insert}
\end{figure}
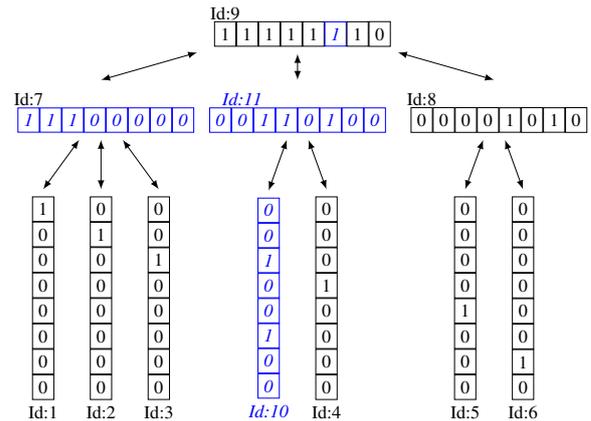

\textbf{Example.} Consider the Bloofi tree in Fig.~\ref{fig:bloofiTree} and assume that we insert the Bloom filter with value ``00100100''. The new node is inserted as a child of node 7 (Hamming distance between the new node and nodes 7 and 8 is 4, so let's assume that node 7 is chosen as the closest node), which needs to split. The resulting Bloofi tree is shown in Fig.~\ref{fig:bloofiTree_insert}.

Theorem~\ref{theorem:insertCost} gives the cost of the insert algorithm for Bloofi. The cost metric we use to measure the performance of an operation in Bloofi is the number of Bloofi nodes accessed by that operation: either the Bloom filter value is read/modified by that operation, or the parent or children pointers in the node are read/modified.

\begin{theorem}[Insert Cost]\label{theorem:insertCost}
The number of Bloofi nodes accessed during the insert operation in Bloofi is $O(d \log_{d} N)$, where $d$ is the order of the Bloofi index and $N$ is the number of Bloom filters that are indexed.
\end{theorem}

\begin{proof}
The following components are part of an insert operation:
\begin{enumerate}
\item  The values of all nodes in the path from the root to the new leaf are updated to reflect the newly inserted Bloom filter. Updating the value of a node means computing the OR with the newly inserted Bloom filter, so only the old and the new Bloom filter values need to be accessed. The cost of the update is therefore constant (2). The height of Bloofi tree is at most $\lfloor log_{d} N \rfloor$ and at each level we perform a constant amount of work, so the total cost for that update is $O(\log_{d} N)$.
\item  At each level in the tree, a search for the most similar child node is performed (line 9), and the cost of that search is $O(d)$, so total cost due to the search for the placement of the new leaf is $O(d \log_d N)$.
\item  The cost of a split is $O(d)$ since there are at most $2d+1$ children for a node. In the worst case, the split propagates up to the root, and the height of the tree is $O(\log_d N)$, so the worst case cost for the split operations is $O(d \log_d N)$.
\end{enumerate}
 From the three points, we see that the cost of the insert operation is $O(d \log_d N)$.\end{proof}

\subsection{Delete}\label{sec:delete}

The delete algorithm (Algorithm~\ref{alg:delete}) deletes the given node from the Bloofi index. When the procedure is first invoked with the leaf to be deleted as argument, the pointer to that leaf in the parent is deleted. If the parent node is not underflowing (at least $d$ children are left), the Bloom filter values of the nodes in the path from the parent node to the root are re-computed to be the bitwise OR of their remaining children (line~\ref{line:recompute}) and the delete procedure terminates. If there is an underflow (lines~\ref{line:underflowstart}--\ref{line:underflowfinish}), the parent node tries to redistribute its entries with a sibling. If redistribution is possible, the entries are re-distributed, parent information in the moving nodes is updated, and the Bloom filter values in the parent node, sibling node, and all the way up to the root are updated (lines~\ref{line:redisstart}--\ref{line:redisfinish}). If redistribution is not possible, the parent node merges with a sibling, by giving all its entries to the sibling (lines~\ref{line:noredisstart}--\ref{line:noredisfinish}). The Bloom filter value in the sibling node is updated, and the delete procedure is called recursively for the parent node. Occasionally, the delete propagates up to the root. If only one child remains in the root, the root is deleted and the height of the tree decreases (lines~\ref{line:rootstart}--\ref{line:rootfinish}).

\textbf{Example.} Assume that the node with id 5 and value ``00001000'' is deleted from Fig.~\ref{fig:bloofiTree}. The resulting tree, after deletion of node 5 and redistribution between 8 and 7 is shown in Fig.~\ref{fig:bloofiTree_delete}.

\begin{theorem}[Delete Cost]\label{theorem:deleteCost}
The number of Bloofi nodes accessed during the delete operation in Bloofi is $O(d \log_{d} N)$, where $d$ is the order of the Bloofi index and $N$ is the number of Bloom filters that are indexed.
\end{theorem}

\begin{proof}
\begin{inparaenum}[(1)]\item Once the reference to a node is deleted from its parent in the Bloofi tree (constant cost operation), the values of the nodes from the deleted node to the root need to be recomputed, so the total cost is $O(d \log_d N)$. \item Occasionally, a redistribute or merge is needed, with a cost in $O(d)$. In the worst case, the merge and delete propagates up to the root, so the worst case cost for merge is $O(d \log_d N)$.
\end{inparaenum}
From (1) and (2) it follows that the cost of the delete operation is $O(d \log_d N)$.
\end{proof}

When subjected to  many insertions and
deletions, though the Bloofi tree remains balanced, it could be that
the partition of the Bloom filters as per their Hamming distance could
degrade in quality. The performance of Bloofi could diminish in such
instances and it could become necessary to reconstruct the Bloofi data structure.

\begin{algorithm}[htbp]
\caption{: \texttt{delete}($\text{childNode}$)}\label{alg:delete}
\begin{algorithmic}[1]
\STATE //find the parent node
\STATE $\text{parentNode} = \text{childNode.parent}$;
\STATE //remove the reference to the node from its parent
\STATE $\text{parentNode.children.remove}(\text{childNode})$;
\STATE //check whether the tree height needs to be reduced
\IF{$\text{parentNode} = \text{root}$ AND $\text{parentNode.nbDesc} = 1$}
\STATE $\text{root} = \text{parentNode.children.get}(0)$;~\label{line:rootstart}
\STATE $\text{root.parent} = \text{null}$;
\RETURN $\text{null}$~\label{line:rootfinish}
\ENDIF
\STATE //if not, check if underflow at the parent
\IF{$\text{parentNode.underflow}$} \label{line:underflowstart}
\STATE //underflow, so try to redistribute first
\STATE $\text{sibling} = \text{sibling of parentNode}$;
\IF{$\text{sibling}.\text{canRedistribute}$}
\STATE //redistribute with sibling\label{line:redisstart}
\STATE remove some children from $\text{sibling}$ to even out the number of children
\STATE insert new children into $\text{parentNode}$
\STATE update $.\text{parent}$ information for all nodes moved
\STATE //update value of all nodes involved, up to the root
\STATE $\text{sibling.val}=$  OR of all children value;
\STATE \texttt{recomputeValueToTheRoot}$(\text{parentNode})$;\label{line:redisfinish}
\ELSE
\STATE //merge with sibling\label{line:noredisstart}
\STATE move all children from $\text{parentNode}$ to $\text{sibling}$;
\STATE update $.\text{parent}$ information for all nodes moved
\STATE //recompute $\text{sibling}$ value
\STATE $\text{sibling.val} = $OR of all childrenValue
\STATE //delete the $\text{parentNode}$
\STATE \texttt{delete}($\text{parentNode}$);\label{line:noredisfinish} \label{line:underflowfinish}
\ENDIF //merge or redistribute
\ELSE
\STATE //no underflow
\STATE \label{line:recompute} //re-compute the value of all Bloom filters up to the root
\STATE \texttt{recomputeValueToTheRoot}$(\text{parentNode})$;
\ENDIF
\end{algorithmic}
\end{algorithm}

\begin{figure}[htbp]
\centering
\begin{tikzpicture}[edge from parent/.style={draw,latex-latex},
scale=0.9, every node/.style={scale=0.9},
start chain=1 going below,
start chain=2 going below,
start chain=3 going below,
start chain=4 going below,
start chain=6 going below,
start chain=7 going right,
start chain=8 going right,
start chain=9 going right,
child anchor=north,
level distance=1.5cm,
level 1/.style={sibling distance=42mm},
level 2/.style={sibling distance=10mm,level distance=3cm}]
\edef\sizetape{0.2cm}
\node{%
\begin{tikzpicture}[node distance=0cm]
    \tikzstyle{mytape}=[draw,minimum size=\sizetape]
    \node (id9) [on chain=9,mytape,blue,font=\itshape] {1};
    \node [on chain=9,mytape,blue,font=\itshape] {1};
    \node [on chain=9,mytape,blue,font=\itshape] {1};
    \node [on chain=9,mytape,blue,font=\itshape] {1};
    \node [on chain=9,mytape,blue,font=\itshape] {0};
        \node [above of=id9,node distance=0.4cm] {Id:9};
    \node [on chain=9,mytape,blue,font=\itshape] {0};
    \node [on chain=9,mytape,blue,font=\itshape] {1};
    \node [on chain=9,mytape,blue,font=\itshape] {0};
\end{tikzpicture}
}
child{node {\begin{tikzpicture}[node distance=0cm]
    \tikzstyle{mytape}=[draw,minimum size=\sizetape]
    \node (id7) [on chain=7,mytape,blue,font=\itshape] {1};
      \node [left =of id7,node distance=0.6cm] {Id:7};
    \node [on chain=7,mytape,blue,font=\itshape] {1};
    \node [on chain=7,mytape,blue,font=\itshape] {1};
    \node [on chain=7,mytape,blue,font=\itshape] {0};
    \node [on chain=7,mytape,blue,font=\itshape] {0};
    \node [on chain=7,mytape,blue,font=\itshape] {0};
    \node [on chain=7,mytape,blue,font=\itshape] {0};
    \node [on chain=7,mytape,blue,font=\itshape] {0};
\end{tikzpicture}
}
child {node{\begin{tikzpicture}[node distance=0cm]
    \tikzstyle{mytape}=[draw,minimum size=\sizetape]
    \node [on chain=1,mytape] {1};
    \node [on chain=1,mytape] {0};
    \node [on chain=1,mytape] {0};
    \node [on chain=1,mytape] {0} ;
    \node [on chain=1,mytape] {0};
    \node [on chain=1,mytape] {0};
    \node [on chain=1,mytape] {0};
    \node (id1)  [on chain=1,mytape] {0};
      \node [below =of id1,node distance=0.7cm] {Id:1};
\end{tikzpicture}}} child {node{\begin{tikzpicture}[node distance=0cm]
    \tikzstyle{mytape}=[draw,minimum size=\sizetape]
    \node [on chain=2,mytape] {0};
         \node [on chain=2,mytape] {1};
    \node [on chain=2,mytape] {0};
    \node [on chain=2,mytape] {0} ;
    \node [on chain=2,mytape] {0};
    \node [on chain=2,mytape] {0};
    \node [on chain=2,mytape] {0};
    \node (id2) [on chain=2,mytape] {0};
    \node [below =of id2,node distance=0.7cm] {Id:2};
\end{tikzpicture}}} child {node{\begin{tikzpicture}[node distance=0cm]
    \tikzstyle{mytape}=[draw,minimum size=\sizetape]
    \node [on chain=3,mytape] {0};
         \node [on chain=3,mytape] {0};
    \node [on chain=3,mytape] {1};
    \node [on chain=3,mytape] {0} ;
    \node [on chain=3,mytape] {0};
    \node [on chain=3,mytape] {0};
    \node [on chain=3,mytape] {0};
    \node (id3) [on chain=3,mytape] {0};
    \node [below =of id3,node distance=0.7cm] {Id:3};
\end{tikzpicture}}} }
child{node {\begin{tikzpicture}[node distance=0cm]
    \tikzstyle{mytape}=[draw,minimum size=\sizetape]
    \node [on chain=8,mytape,blue, font=\itshape] {0};
    \node [on chain=8,mytape,blue, font=\itshape] {0};
    \node [on chain=8,mytape,blue, font=\itshape] {0};
    \node [on chain=8,mytape,blue, font=\itshape] {1} ;
    \node [on chain=8,mytape,blue, font=\itshape] {0};
    \node [on chain=8,mytape,blue, font=\itshape] {0};
    \node [on chain=8,mytape,blue, font=\itshape] {1};
    \node (id8) [on chain=8,mytape,blue, font=\itshape] {0};
     \node  [right of= id8,node distance=0.6cm] {Id:8};
\end{tikzpicture}}
child {node{\begin{tikzpicture}[node distance=0cm]
    \tikzstyle{mytape}=[draw,minimum size=\sizetape]
    \node [on chain=4,mytape] {0};
         \node [on chain=4,mytape] {0};
    \node [on chain=4,mytape] {0};
    \node [on chain=4,mytape] {1} ;
    \node [on chain=4,mytape] {0};
    \node [on chain=4,mytape] {0};
    \node [on chain=4,mytape] {0};
    \node (id4) [on chain=4,mytape] {0};
    \node [below =of id4,node distance=0.6cm] {Id:4};
\end{tikzpicture}}}
child {node{\begin{tikzpicture}[node distance=0cm]
    \tikzstyle{mytape}=[draw,minimum size=\sizetape]
    \node [on chain=6,mytape] {0};
         \node [on chain=6,mytape] {0};
    \node [on chain=6,mytape] {0};
    \node [on chain=6,mytape] {0} ;
    \node [on chain=6,mytape] {0};
    \node [on chain=6,mytape] {0};
    \node [on chain=6,mytape] {1};
    \node (id6) [on chain=6,mytape] {0};
    \node [below =of id6,node distance=0.7cm] {Id:6};
\end{tikzpicture}}} };
\end{tikzpicture}

\caption{Bloofi Tree After Delete and Redistribute}
\label{fig:bloofiTree_delete}
\end{figure}
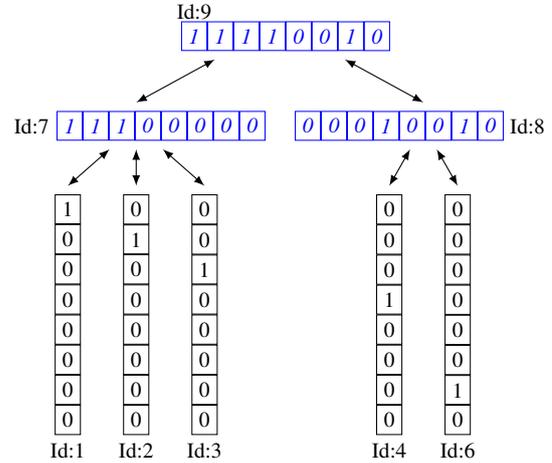

\begin{algorithm}[htbp]
\caption{: \texttt{update}(leaf,newBloomFilter)}\label{alg:update}
\begin{algorithmic}[1]
\STATE //update all values on the path from leaf to the root
\STATE $\text{node} = \text{leaf}$;
\REPEAT
\STATE $\text{node.val} = \text{node.val} \, \big | \, \text{newBloomFilter}$;
\STATE $\text{node} = \text{node.parent}$;
\UNTIL{$\text{node} = \text{null}$}
\end{algorithmic}
\end{algorithm}

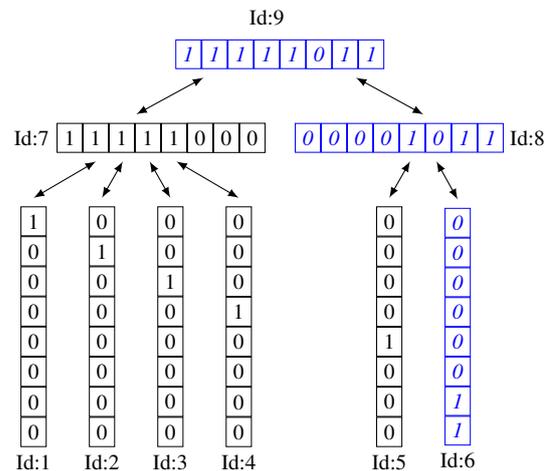
\begin{figure}[htbp]
\centering
\begin{tikzpicture}[edge from parent/.style={draw,latex-latex},
scale=0.9, every node/.style={scale=0.9},
start chain=1 going below,
start chain=2 going below,
start chain=3 going below,
start chain=4 going below,
start chain=5 going below,
start chain=6 going below,
start chain=7 going right,
start chain=8 going right,
start chain=9 going right,
child anchor=north,
level distance=1.5cm,
level 1/.style={sibling distance=42mm},
level 2/.style={sibling distance=10mm,level distance=3cm}]
\edef\sizetape{0.2cm}
\node{%
\begin{tikzpicture}[node distance=0cm]
    \tikzstyle{mytape}=[draw,minimum size=\sizetape]
    \node  [on chain=9,mytape,blue,font=\itshape] {1};
    \node [on chain=9,mytape,blue,font=\itshape] {1};
    \node [on chain=9,mytape,blue,font=\itshape] {1};
    \node(id9) [on chain=9,mytape,blue,font=\itshape] {1};
    \node [on chain=9,mytape,blue,font=\itshape] {1};
        \node [above of=id9,node distance=0.6cm] {Id:9};
    \node [on chain=9,mytape,blue,font=\itshape] {0};
    \node [on chain=9,mytape,blue,font=\itshape] {1};
    \node [on chain=9,mytape,blue,font=\itshape] {1};
\end{tikzpicture}
}
child{node {\begin{tikzpicture}[node distance=0cm]
    \tikzstyle{mytape}=[draw,minimum size=\sizetape]
    \node (id7) [on chain=7,mytape] {1};
      \node [left =of id7,node distance=0.6cm] {Id:7};
    \node [on chain=7,mytape] {1};
    \node [on chain=7,mytape] {1};
    \node [on chain=7,mytape] {1} ;
    \node [on chain=7,mytape] {1};
    \node [on chain=7,mytape] {0};
    \node [on chain=7,mytape] {0};
    \node [on chain=7,mytape] {0};
\end{tikzpicture}
}
child {node{\begin{tikzpicture}[node distance=0cm]
    \tikzstyle{mytape}=[draw,minimum size=\sizetape]
    \node [on chain=1,mytape] {1};
    \node [on chain=1,mytape] {0};
    \node [on chain=1,mytape] {0};
    \node [on chain=1,mytape] {0} ;
    \node [on chain=1,mytape] {0};
    \node [on chain=1,mytape] {0};
    \node [on chain=1,mytape] {0};
    \node (id1)  [on chain=1,mytape] {0};
      \node [below =of id1,node distance=0.7cm] {Id:1};
\end{tikzpicture}}} child {node{\begin{tikzpicture}[node distance=0cm]
    \tikzstyle{mytape}=[draw,minimum size=\sizetape]
    \node [on chain=2,mytape] {0};
         \node [on chain=2,mytape] {1};
    \node [on chain=2,mytape] {0};
    \node [on chain=2,mytape] {0} ;
    \node [on chain=2,mytape] {0};
    \node [on chain=2,mytape] {0};
    \node [on chain=2,mytape] {0};
    \node (id2) [on chain=2,mytape] {0};
    \node [below =of id2,node distance=0.7cm] {Id:2};
\end{tikzpicture}}} child {node{\begin{tikzpicture}[node distance=0cm]
    \tikzstyle{mytape}=[draw,minimum size=\sizetape]
    \node [on chain=3,mytape] {0};
         \node [on chain=3,mytape] {0};
    \node [on chain=3,mytape] {1};
    \node [on chain=3,mytape] {0} ;
    \node [on chain=3,mytape] {0};
    \node [on chain=3,mytape] {0};
    \node [on chain=3,mytape] {0};
    \node (id3) [on chain=3,mytape] {0};
    \node [below =of id3,node distance=0.7cm] {Id:3};
\end{tikzpicture}}} child {node{\begin{tikzpicture}[node distance=0cm]
    \tikzstyle{mytape}=[draw,minimum size=\sizetape]
    \node [on chain=4,mytape] {0};
         \node [on chain=4,mytape] {0};
    \node [on chain=4,mytape] {0};
    \node [on chain=4,mytape] {1} ;
    \node [on chain=4,mytape] {0};
    \node [on chain=4,mytape] {0};
    \node [on chain=4,mytape] {0};
    \node (id4) [on chain=4,mytape] {0};
    \node [below =of id4,node distance=0.6cm] {Id:4};
\end{tikzpicture}}} }
child{node {\begin{tikzpicture}[node distance=0cm]
    \tikzstyle{mytape}=[draw,minimum size=\sizetape]
    \node [on chain=8,mytape,blue,font=\itshape] {0};
    \node [on chain=8,mytape,blue,font=\itshape] {0};
    \node [on chain=8,mytape,blue,font=\itshape] {0};
    \node [on chain=8,mytape,blue,font=\itshape] {0} ;
    \node [on chain=8,mytape,blue,font=\itshape] {1};
    \node [on chain=8,mytape,blue,font=\itshape] {0};
    \node [on chain=8,mytape,blue,font=\itshape] {1};
    \node (id8) [on chain=8,mytape,blue,font=\itshape] {1};
     \node  [right of= id8,node distance=0.6cm] {Id:8};
\end{tikzpicture}} child {node{\begin{tikzpicture}[node distance=0cm]
    \tikzstyle{mytape}=[draw,minimum size=\sizetape]
    \node [on chain=5,mytape] {0};
         \node [on chain=5,mytape] {0};
    \node [on chain=5,mytape] {0};
    \node [on chain=5,mytape] {0} ;
    \node [on chain=5,mytape] {1};
    \node [on chain=5,mytape] {0};
    \node [on chain=5,mytape] {0};
    \node (id5) [on chain=5,mytape] {0};
    \node [below =of id5,node distance=0.7cm] {Id:5};
\end{tikzpicture}}}  child {node{\begin{tikzpicture}[node distance=0cm]
    \tikzstyle{mytape}=[draw,minimum size=\sizetape]
    \node [on chain=6,mytape,blue,font=\itshape] {0};
    \node [on chain=6,mytape,blue,font=\itshape] {0};
    \node [on chain=6,mytape,blue,font=\itshape] {0};
    \node [on chain=6,mytape,blue,font=\itshape] {0} ;
    \node [on chain=6,mytape,blue,font=\itshape] {0};
    \node [on chain=6,mytape,blue,font=\itshape] {0};
    \node [on chain=6,mytape,blue,font=\itshape] {1};
    \node (id6) [on chain=6,mytape,blue,font=\itshape] {1};
    \node [below =of id6,node distance=0.7cm] {Id:6};
\end{tikzpicture}}} };
\end{tikzpicture}
\caption{Bloofi Tree After Update}
\label{fig:bloofiTree_update}
\end{figure}

\subsection{Update}\label{sec:update}

Object insertions in the underlying set lead to updates of the Bloom filters, so we expect the update operation for Bloofi to be quite frequent. Instead of treating a Bloom filter update as a delete followed by insert, we use an ``in-place'' update. If the Bloofi tree becomes inefficient in routing due to updates (too many false positives during search) the Bloofi tree can be reconstructed from scratch in batch mode.
Algorithm~\ref{alg:update} shows the pseudo-code for the update algorithm. The algorithm takes as parameters the leaf node corresponding to the updated value and the new Bloom filter value for that node. All the Bloom filters in the path from the leaf to the root are updated by OR-ing with the new value.

\textbf{Example.} In the Bloofi tree in Fig.~\ref{fig:bloofiTree}, assume that we update the value of node 6 to be ``00000011''. The values of all the nodes in the path from node 6 to the root are OR-ed with ``00000011'' and the resulting tree is shown in Fig.~\ref{fig:bloofiTree_update}.

\begin{theorem}[Update Cost]\label{theorem:updateCost}
The number of Bloofi nodes accessed during the update operation in Bloofi is $O(\log_{d} N)$, where $d$ is the order of the Bloofi index and $N$ is the number of Bloom filters that are indexed.
\end{theorem}

\subsection{Improving pruning efficiency}\label{sec:noSplitFull}

In Bloofi, each non-leaf node value is the bitwise OR of its children values. As the total number of objects in the underlying sets indexed by Bloofi increases, the probability of false positive results returned by the Bloom filters at the higher levels in the tree increases. In the worst case, all bits in the Bloofi nodes at higher levels in the tree could be one. This leads to decreased pruning efficiency of the higher levels in the tree, as more false positive paths are followed during search. To improve upon the number of Bloom filters that need to be checked for matches during a query, we propose the following heuristic.

During the insert procedure, we do not split a node that has all the bits set to one, even if the node is full. This could stop the splitting a little too early, but it avoids creating multiple levels in the tree with all the bits set to one. Our experimental results in Section~\ref{sec:performance} show that the search cost is indeed improved by using this heuristic when the root level value has all bits set to one.
Effectively, Bloofi can be viewed as a \emph{forest} instead of a tree.

Alternatively, we could dynamically change the size of the Bloom filters when the false positive probability at the root reaches 1. In such a case, if the application allows, we could reconstruct the base Bloom filters to have a lower false positive probability, and reconstruct the Bloofi tree from bottom-up. The index construction time for \num{100000}~Bloom filters was only about 15 seconds in our experiments, so periodic reconstruction of the index is a viable solution.

\section{Bit-level parallelism}
\label{sec:flat}

Bloofi keeps the Bloom filters as they are, and only adds new (aggregated) Bloom filters to accelerate the queries. However, in a worst case scenario, Bloofi may need to check many Bloom filters. In such a case, Bloofi may not be faster, and could even be slower, than a naive approach which merely checks every Bloom filter.

Checking for membership in a Bloom filter is equivalent to checking the value of a few bits at random locations in a bitmap. Though fast, this operation does not exploit bit-level parallelism: the processor's ability to do several bitwise operations in one instruction.

Let us assume that we have a 64-bit processor.
We propose a new approach (Flat-Bloofi) which stores the data corresponding to 64~Bloom filters in a packed data structure.
Each Bloom filter is backed by a $m$-bit bitmap. In their place, we construct a single array of 64-bit integers of length $m$ (henceforth a Flat-Bloofi array).
 The first 64-bit integer corresponds to the first bit of each of the 64~bitmaps. And so on. Thus, the value of the $i^{\mathrm{th}}$~bit of the $j^{\mathrm{th}}$~bitmap is the value of the $j^{\mathrm{th}}$~bit of the $i^{\mathrm{th}}$~integer in Flat-Bloofi array. Given $N$~Bloom filters, we create $\lceil N / 64 \rceil$~Flat-Bloofi arrays.
 When $N$ is not a multiple of 64, some bits are unused in one of the Flat-Bloofi arrays.

We organize Flat-Bloofi using the following data structures:
\begin{itemize}
\item We maintain $\zeta$~Flat-Bloofi arrays.
With this data, we can index $L=\zeta \times 64$~Bloom filters.

\item We use an array $\beta$ of  $L$~bits, of which exactly $N$ are set to true. This indicates which index locations are in use. Thus if a Bloom filter is deleted from the index and then a new one inserted, we can reuse the space.
\item We use a hash table which maps from Bloom filter identifiers to internal index values (in the range $[0,L)$).
We also maintain an array of identifiers of length up to $L$ which gives us the identifier of the Bloom filter \emph{stored} at a given index. Combined together, the hash table and the array of identifiers provide
a two-way index from Bloom filter identifiers to index locations.
\end{itemize}

\paragraph{Queries} Given $k$~hash functions, we map a given value to $k$~index locations in the range $[0,m)$. For each Flat-Bloofi array, we retrieve the corresponding $k$~64-bit integers and compute their bit-wise AND aggregate.
We then iterate over the bits having a value of true: each one corresponds to a matching Bloom filter.
We use our array of identifiers to recover the corresponding Bloom filter identifiers. Thus if we have
$N$~Bloom filters, we will access no more than $k \times \zeta $~64-bit integers from the Flat-Bloofi arrays. Iterating through the set bits can be done quickly using fast functions such as Java's \texttt{Long.bitCount}.

\paragraph{Insertion} When inserting a new Bloom filter, using the bit array $\beta$, we first seek an available index. If none is found, then we create a new array of 64-bit integers, thus, effectively, making available 64~new index positions. The bit array $\beta$, the array of identifiers, and the hash table mapping to index values are updated. Finally, we iterate through all of the set bits in the bitmap of the new Bloom filter and set the corresponding bits in the Flat-Bloofi array using a bit-wise OR operation.

\paragraph{Deletion}   When deleting a Bloom filter,  we use the hash table to recover
its index using its identifier. The key is then removed from the hash table.
The corresponding bit in  $\beta$ is set to false.
There are then two possibilities.
\begin{itemize}
\item  If this Bloom filter was stored alone in a Flat-Bloofi array,
then the Flat-Bloofi array is removed. We also remove the corresponding 64~bits in  $\beta$ as well
as the corresponding 64~entries in the array of identifiers. We scan the values  of the
hash table, and deduct 64 from all index entries exceeding index of the deleted Bloom filter.
\item Otherwise, we go through the Flat-Bloofi array and unset the bits corresponding
to the Bloom filter using a bit-wise AND operation. Because we do not
keep a copy of the original Bloom filter, we need to update every single
component of the Flat-Bloofi array.
\end{itemize}
We believe that this compaction approach should provide reasonable performance and memory usage in a context where deletions and insertions are frequent. In the worst case scenario, however, and after many deletions, we could have $N$~Flat-Bloofi arrays indexing $N$~Bloom filters. To guard against such inefficiencies, we would need more aggressive compaction strategies, but we leave them to future work.

Note
We could further accelerate the queries by replacing many  ($\zeta$) Flat-Bloofi arrays, with  a single array containing words of $\zeta \times 64$~bits. This would improve memory locality. However, it would make compaction more expensive.

\paragraph{Update} Updating a Bloom filter is easiest: we  go through the corresponding
Flat-Bloofi array and set the corresponding bits as we did with an insertion.

\section{Experimental evaluation}\label{sec:experiments}
We evaluate Bloofi's and Flat-Bloofi's search performance and maintenance cost for different number and size of Bloom filters, different underlying data distributions, and, for Bloofi, different similarity metrics used for Bloofi construction and different order values. We also compare Bloofi and Flat-Bloofi's performance with the ``naive'' case, where the Bloom filters are searched linearly, without any index. We show that in most cases, Bloofi achieves logarithmic search performance, with low maintenance cost, regardless of the underlying data distribution.

\subsection{Experiments setup}
We implemented Bloofi and Flat-Bloofi in Java and ran our experiments using  Oracle's JDK version 1.7.0\_45. For Bloom filters, we use the implementation provided by Skjegstad~\cite{bloomFilterMagnus}, modified to use  faster hashing and a more efficient BitSet implementation~\cite{bitSetLemire}. The experiments were run on a HP~Z820 Workstation with an Intel Xeon E5-2640 processor (2.50\,GHz) having 6~cores. Our test machine had
 32\,GB of RAM (DDR3-1600, $4\times 8$\,GB). We did not use parallelism and all our data structures are in RAM\@.

 \subsubsection{Performance metrics}
 As performance metrics we use:
  \begin{itemize}
  \item \emph{search bf-cost}: the number of Bloom filters checked to find the one(s) matching a queried object, averaged over 50\,000~searches
  \item \emph{search time}: the average time, in milliseconds, to find all the Bloom filters matching a queried object
  \item \emph{storage cost}: space required to store the Bloom filters and/or the associated index structure, if applicable. For Bloofi, we estimate this cost as the number of bytes for a Bloom filter, multiplied by the number of nodes in the Bloofi tree, including the leaves; for Flat-Bloofi, the storage cost is estimated as the number of bytes for a Bloom filters, multiplied by number of Bloom filters rounded up to a multiple of 64 (Flat-Bloofi uses longs for storage); for the ``naive'' case, the storage cost is estimated as the number of bytes for a Bloom filters, multiplied by number of Bloom filters
  \item \emph{maintenance bf-cost}: the average number of Bloofi nodes accessed during an insert, delete, or update operation
  \item \emph{maintenance time}: the average time, in milliseconds, for an insert, delete, or update operation
 \end{itemize}

\subsubsection{Parameters varied}
As described in Section~\ref{sec:bloom}, a Bloom filter is a bit array of length $m$ constructed using a set of $k$ hash functions. In practice, when constructing a Bloom filter, one does not specify the length and the number of hash functions. It is unlikely that the average engineer would know which values to pick. Rather,
one specifies the expected maximal number of elements to be stored, $\nbElExpected$, (a possibly large number) as well as the desired probability of having a false positive, $\setprobf$, (e.g., 1\%). One can reasonably expect an engineer to be able to set these values from domain knowledge. Then the number of hash functions and the size of the bitmaps can be computed using the following formulas $k = \left \lceil{- \ln \setprobf / \ln 2} \right \rceil$ and $m=\left \lceil{k / \ln 2 * \nbElExpected} \right \rceil$.

In the experiments we vary the following parameters:
\begin{itemize}
\item $N$: the number of Bloom filters indexed;
\item $d$: Bloofi order;
\item $m$: the size, in bits, of the Bloom filters indexed: $m$ is varied indirectly, by specifying $\nbElExpected$;
\item $n$: the number of elements in each Bloom filter indexed;
\item $\setprobf{}$: the desired probability of false positives in the Bloom filters indexed;
\item \emph{index construction method}: $iterative$, where we insert Bloom filters one by one using the algorithm in Section~\ref{sec:insert}, or $bulk$. For the bulk construction,
we first sort all Bloom filters such that the first Bloom filter is the one closest to the empty Bloom filter, the second is the filter closest to the first Bloom filter, etc., and then construct the Bloofi tree by always inserting next to the right-most leaf;
\item \emph{similarity measure}: the measure used to define ``closeness'' during insert. We consider Hamming, Cosine and Jaccard distances;
\item \emph{data distribution}: $nonrandom$, with non-overlapping ranges for each Bloom filter: each Bloom filter $i$ contains $\nbEl{}$~integers in the range $[i \times \nbEl{}, (i+1) \times \nbEl{})$ and $random$, with overlapping ranges for data in the Bloom filters: each Bloom filter $i$ contains $\nbEl{}$ random integers in a randomly assigned range.
\end{itemize}

For each experiment we vary one parameter and use the default values shown in Table~\ref{table:defaultValues} for the rest. We run each experiment 10~times, and we report the averages over the last 5~runs. Since our values $x$ are integers, we picked  hash functions at random of the form $h(x) = a x  \bmod{m}$ where $a$ is an odd integer defining the hash function.  We found that this  choice gave good performance in our case.

\begin{table}
 \centering\small
   \begin{tabular}{ll} \toprule
   parameter & value\\
\midrule        $N$ --- Number of Bloom filters indexed & 1000 \\ 
        $d$ --- Bloofi order & 2 \\ 
        $m$ --- Bloom filter size (bits) & \num{100992} \\ 
        $n$ --- Nb of elements in each Bloom filter & 100 \\
        $\setprobf{}$ --- Desired probability of false positives & 0.01 \\
        Construction method & iterative \\ 
        Similarity measure & Hamming \\ 
        Data distribution & nonrandom \\ 
\bottomrule
    \end{tabular}
\caption{Default Values for Parameters}\label{table:defaultValues}
\end{table}

\subsection{Performance results}\label{sec:performance}

\threefiguresnew
{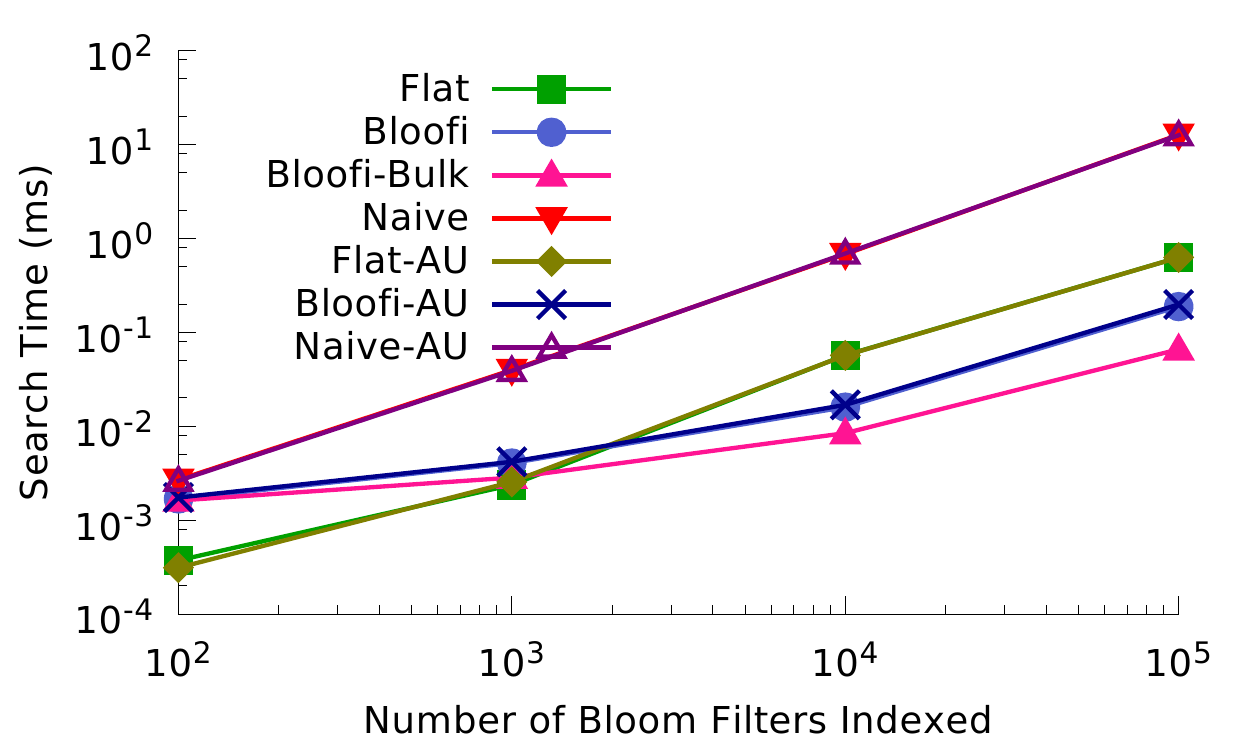}{Search Time  vs.~$N$}{fig:yesSearchTimeVsNbBFs}
{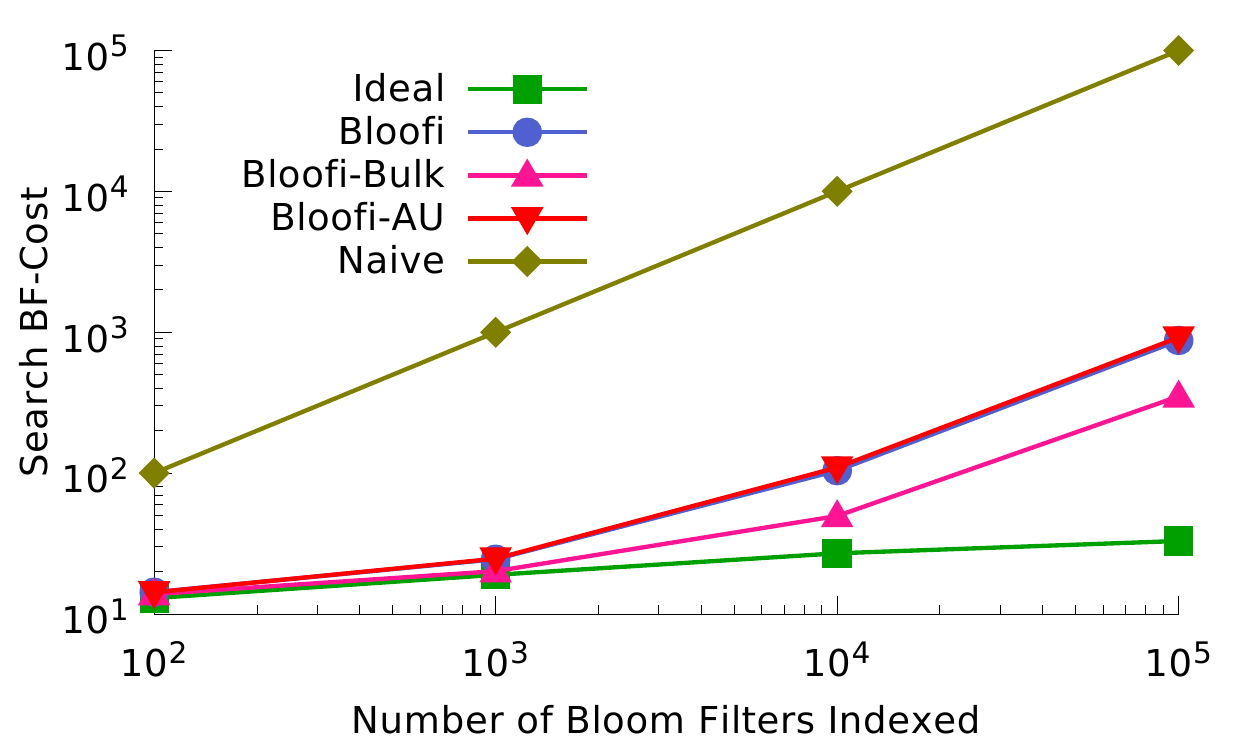}{Search BF-Cost vs.~$N$}{fig:yesSearchesVsNbBFs}
{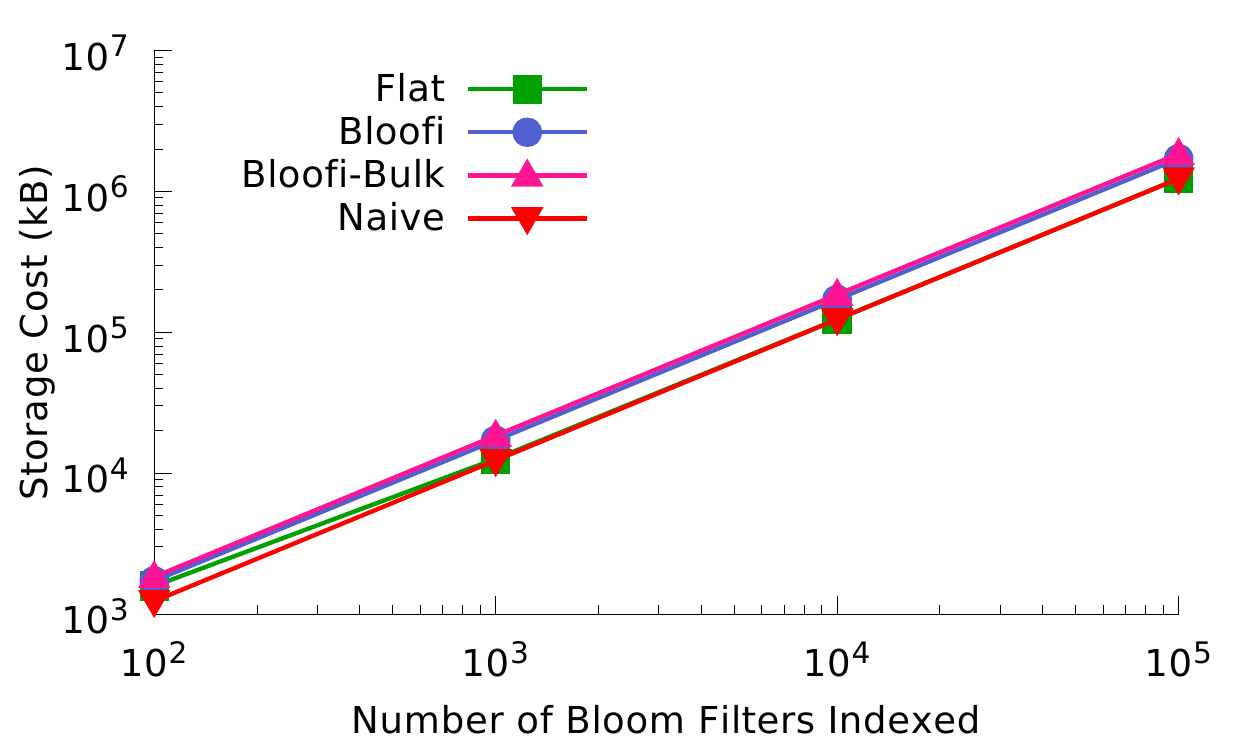}{Storage Cost vs.~$N$}{fig:storageCostVsNbBFs}
{Varying Number of Bloom Filters $N$}

\subsubsection{Varying number of Bloom filters indexed $N$}\label{sec:varyN}

Fig.~\ref{fig:yesSearchTimeVsNbBFs} shows the increase in search time as $N$ increases from 100 to \num{100000} for Bloofi, Flat-Bloofi, and Naive case. Note the logarithmic scale on both $X$ and $Y$ axes. For Bloofi, the increase in search time is logarithmic as long as the false positive probability ($\probf{}$) at the root is less than one ($N\le1000$), but the increase is higher than logarithmic after that, due to high value for $\probf{}$ at the high levels of the tree. However, even for \num{100000}~filters, Bloofi still performs orders of magnitude better than the ``naive'' case. The difference between ``naive'' case and Bloofi is not so big when only a few Bloom filters are checked, likely because the search time depends not only on the number of Bloom filters checked, but also on the locality of nodes in memory. We see that the curves for Bloofi and Flat-Bloofi intersect each other. For small numbers of Bloom filters, Flat-Bloofi performs better, due to its superior memory locality and exploitation of bit-level parallelism. However as the number of Bloom filters increases, Bloofi performs better, due to its superior pruning abilities.

For Bloofi, using bulk construction leads to improved search performance, since a global sort of all the Bloom filters is performed before the bulk insert, while the incremental construction is greedy and might not lead to optimal placement. However, the cost of sorting is $O(N^2)$ in our implementation, which leads to high index construction time for the bulk construction.

To evaluate the effect of the in-place update for Bloofi introduced in Section~\ref{sec:update}, we performed experiments where the Bloofi tree is built incrementally using only half of the elements in each Bloom filter. The rest of the elements are inserted then in each Bloom filter and Bloofi is updated in-place. We perform a similar experiment for Flat-Bloofi and the Naive case. The AU (After Updates) curves in Fig.~\ref{fig:yesSearchTimeVsNbBFs} show the search time in the final data structure. As we expected, for Flat-Bloofi and Naive, the After Updates and normal curves are almost identical, since the properties of these data structures are not affected by updates. However, the Bloofi-AU and Bloofi curves are also almost identical, which shows that the in-place update maintains the search performance of the Bloofi tree.

Fig.~\ref{fig:yesSearchesVsNbBFs} shows similar trends for the search bf-cost (average number of Bloom Filter nodes accessed for a search). The search bf-cost performance metric is not used for Flat-Bloofi and Naive, as for these data strictures, all the Bloom filters are checked during a search. In the ``ideal'' case for Bloofi, when exactly one path from root to the leaf is followed during search, the search bf-cost is approximately $l \log_l N + 1$ if each non-leaf node has $l$ children, and the search bf-cost increases logarithmically with $N$. In our experiments, the increase in search bf-cost is logarithmic as long as the false positive probability ($\probf{}$) at the root is less than one ($N\le1000$), but the increase is higher than logarithmic after that, due to high value for $\probf{}$ at the high levels of the tree. However, even for \num{100000}~filters, Bloofi still performs two orders of magnitude better than the ``naive'' case. If the size of the Bloom filters increases, the search cost decreases to the ``ideal'' cost, as shown in Fig.~\ref{fig:yesSearchesVsFilterSize}.

To evaluate the effects of the heuristic introduces in Section~\ref{sec:noSplitFull}, we run the same experiment without using the heuristic, so always splitting the root even if all the bits were set to 1. When there are Bloofi nodes with all bits 1, both  the search bf-cost and the search time when the heuristic is used are lower than when the heuristic is not used (search bf-cost is 104.29 vs.~110.17 for $N=\num{10000}$ and 876.33 vs.~974.92 for $N=\num{100000}$). This shows that using the heuristic increases the search performance of Bloofi index when the false positive probability at the high levels in the tree is high.

Fig.~\ref{fig:storageCostVsNbBFs} shows the increase in the storage cost for Bloofi, Flat-Bloofi, and Naive, as the number of Bloom filters indexed increases. In all cases, the storage cost increases linearly with $N$. The storage cost is lowest for Naive case, as no extra information besides the actual Bloom filters is maintained. Flat-Bloofi has a small overhead, only because sometimes rounds up the space to multiples of 64, so it can take advantage of the bit-level parallelism. The storage cost for Bloofi is also quite low, as the number of non-leaf nodes in the Bloofi tree is less than $N$. The storage cost for bulk construction (Bloofi-bulk) is slightly higher than for incremental construction (Bloofi), because constructing a tree by always inserting in the right-most leaf leads in general to skinnier trees, with more levels and more nodes.

\threefiguresnew
{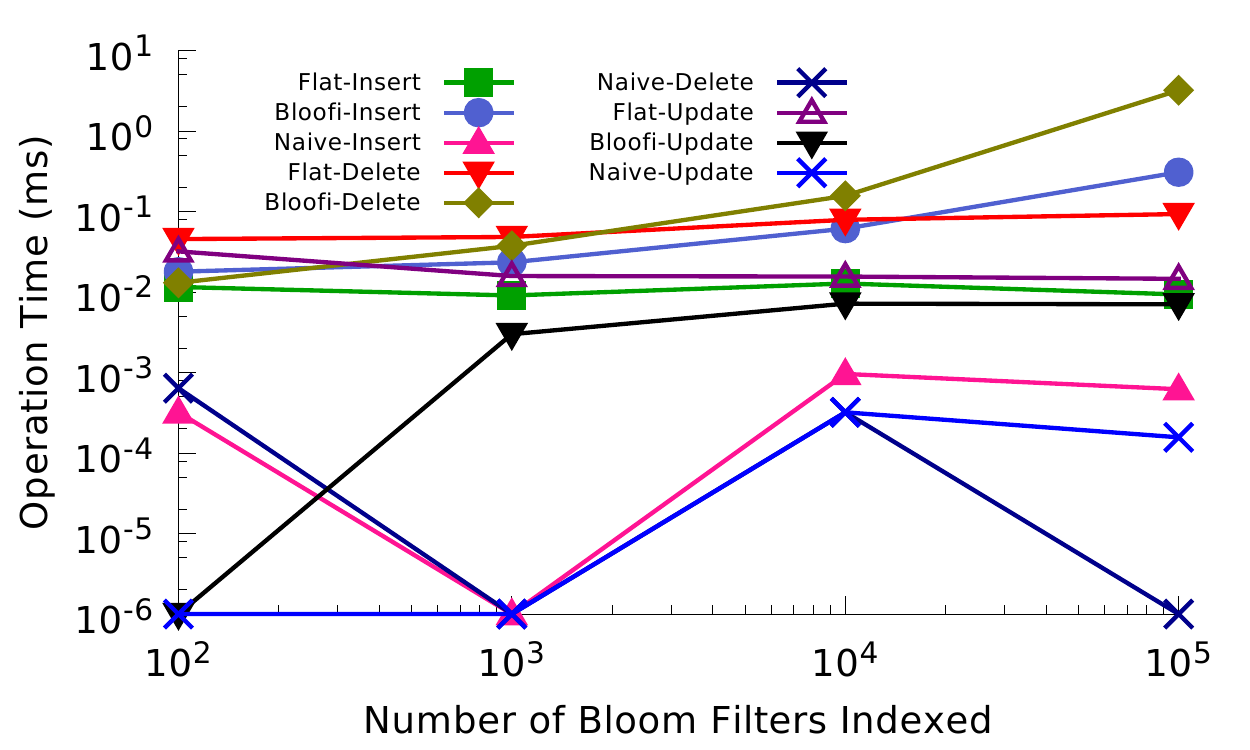}{Maintenance Time vs.~$N$}{fig:insDelTimeVsNbBFs}
{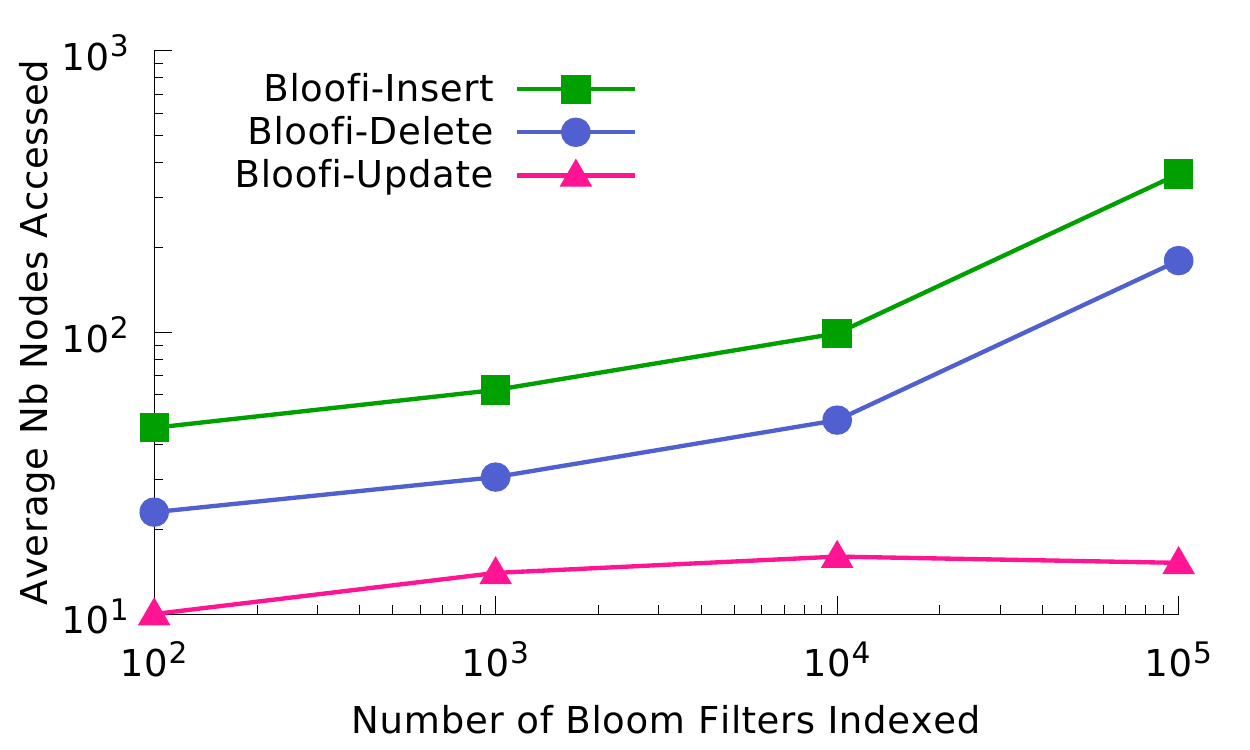}{Maintenance BF-Cost vs.~$N$}{fig:insDelCostVsNbBFs}
{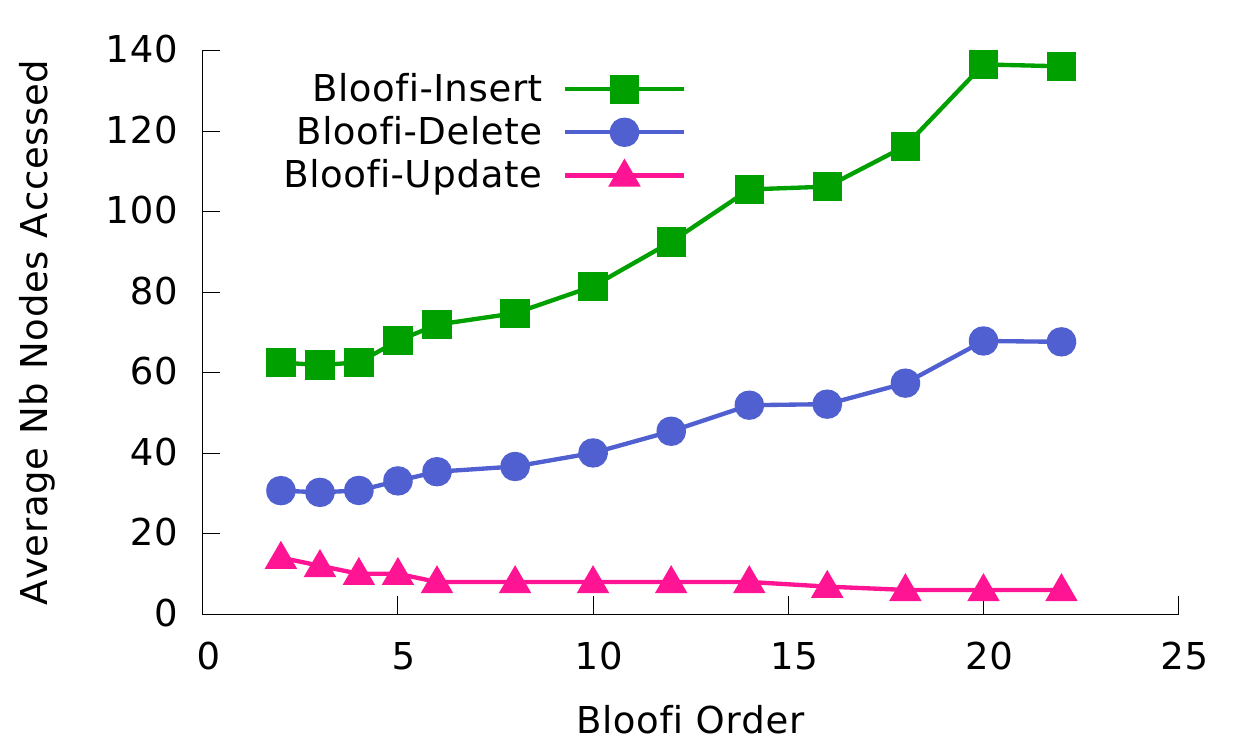}{Maintenance BF-Cost vs.~$d$}{fig:insDelCostVsOrder}
{Maintenance Cost}

Fig.~\ref{fig:insDelTimeVsNbBFs} shows the cost, in terms of time, of maintaining Bloofi, Flat-Bloofi, and the Naive data structure. The maintenance cost for Naive is negligible, as the Bloom filters are just maintained in a list. For Flat-Bloofi, the maintenance cost increases only slightly with the number of Bloom filters, as the cost of the different operations in Flat-Bloofi depends mainly on the size of the Bloom filters and the number of bits turned on, and does not depend much on the number of Bloom filters. For Bloofi, the trend is increasing for the insert and delete, as there are more nodes in the Bloofi tree that get impacted by insert or delete. The cost of updates does not increase after the root becomes all one and does not split, because we use in-place updates and the height of the tree does not increase after $N > 10000$.

We find the relative cost of insert, delete, and update for Bloofi and Flat-Bloofi  quite interesting. For Bloofi, update is the cheapest operation, both as time and number of Bloom filters accessed (see Fig.~\ref{fig:insDelCostVsNbBFs}), since only the values of the nodes from the leaf to the root need to be updated (OR-ed with the new value). However, for Bloofi, the insert operation is the most expensive as bf-cost, but delete becomes the most expensive as time, with the increased number of indexed Bloom filters. The reason for this difference is that while during a delete operation fewer nodes are accessed (no need to search for the place in the tree), we do more work at each node --- we need to re-compute the values of each node in the path from the leaf to root, by OR-ing all the children value. During insert, the values of the nodes get updated by OR-ing with the newly inserted Bloom filter. For Flat-Bloofi, when inserting a sparse Bloom filter, only a few words need to change, so it is  fast. When deleting, we do not know which bit changed, so we have to set them all to zero, which takes longer. If the Bloom filters would be more dense, it is possible that deletions could be faster than insertions. During updates, we do not know how the new Bloom filter differs from the old one, and we cannot compute the difference (with XOR) since we no longer have the old one, so we have to go through the set bits in the Bloom filter, and set them all to 1. The running time is close to an insert, which is what we see in Fig.~\ref{fig:insDelTimeVsNbBFs}. Comparing Bloofi and Flat-Bloofi, Bloofi's updates are cheaper, as Bloofi only needs to OR the new value with a few Bloom filters. For inserts, Flat-Bloofi is faster, as it does not need to search (and compute distances) the tree for the best place to insert the new node. For deletes, Bloofi is faster when the number of Bloom filters is below \num{10000} (and the $\probf{}$ at the root is not 1), but Flat-Bloofi's cost does not depend much on the number of Bloom filters, so it becomes faster for large number of Bloom filters indexed.

The maintenance bf-cost for Bloofi, as the average number of Bloom filters accessed during a maintenance operation, is shown in Fig.~\ref{fig:insDelCostVsNbBFs}. The maintenance bf-cost increases logarithmically with $N$ (Fig.~\ref{fig:insDelCostVsNbBFs}), as expected from Theorems~\ref{theorem:insertCost}--\ref{theorem:updateCost}. The bf-cost of insert is higher than the delete, as the place for the new node needs to be found during insert. The update cost is the lowest, as we use in-place updates, and no splits, merges, or redistributions are needed. The update bf-cost does not increase after the root of the Bloofi tree becomes one, as the root does not split in that case, so the height of the Bloofi tree does not increase with the increased number of indexed Bloom filters.

\threefiguresnew
{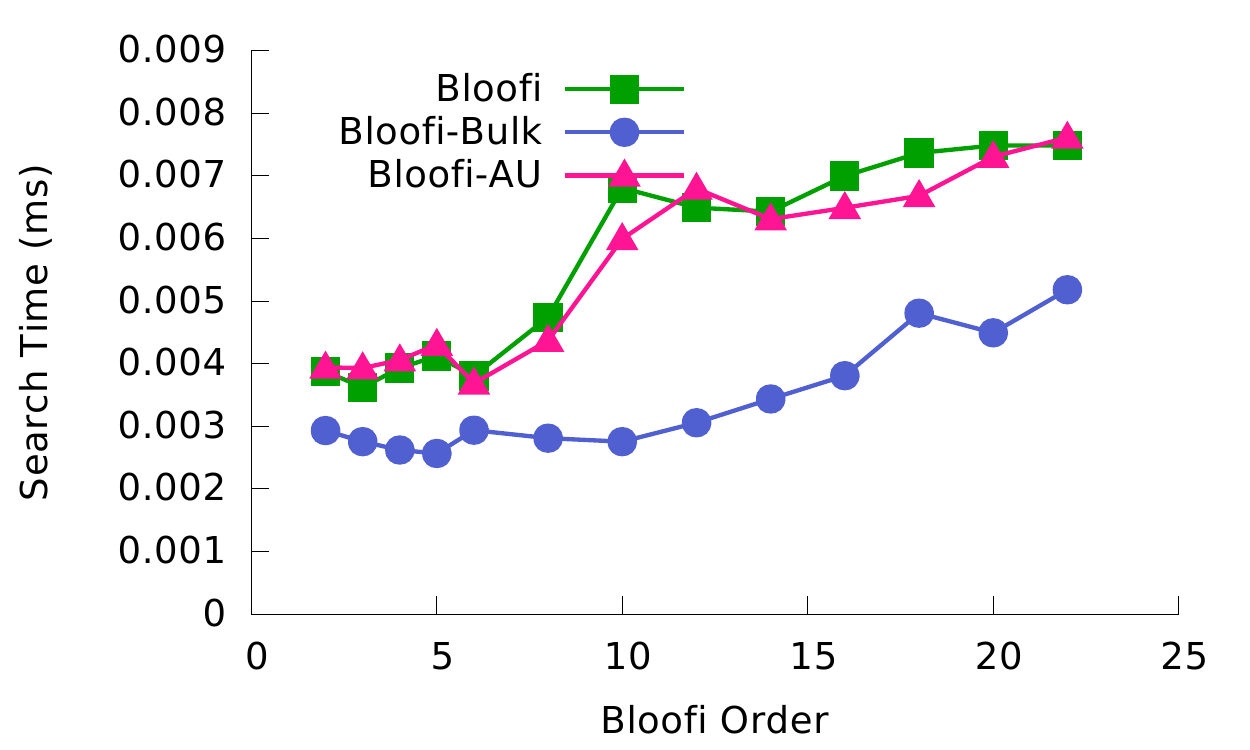}{Search Time vs.~$d$}{fig:yesSearchTimeVsOrder}
{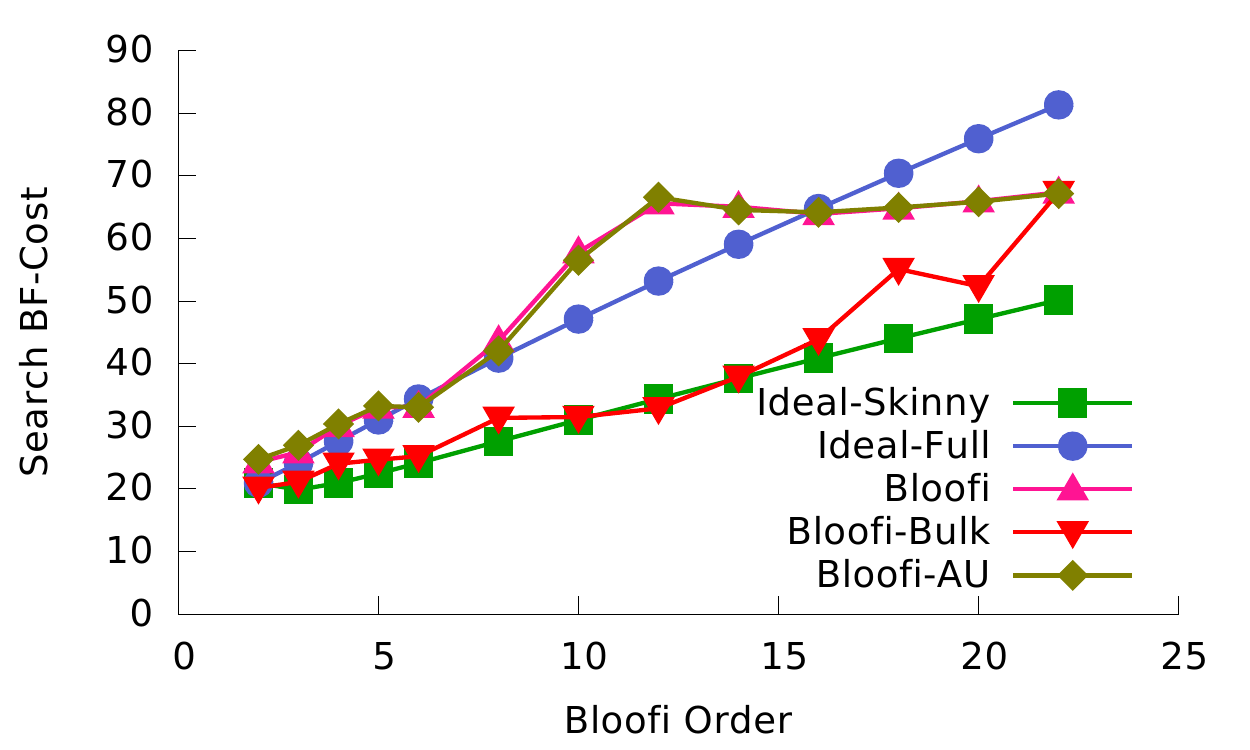}{Search BF-Cost vs.~$d$}{fig:yesSearchesVsOrder}
{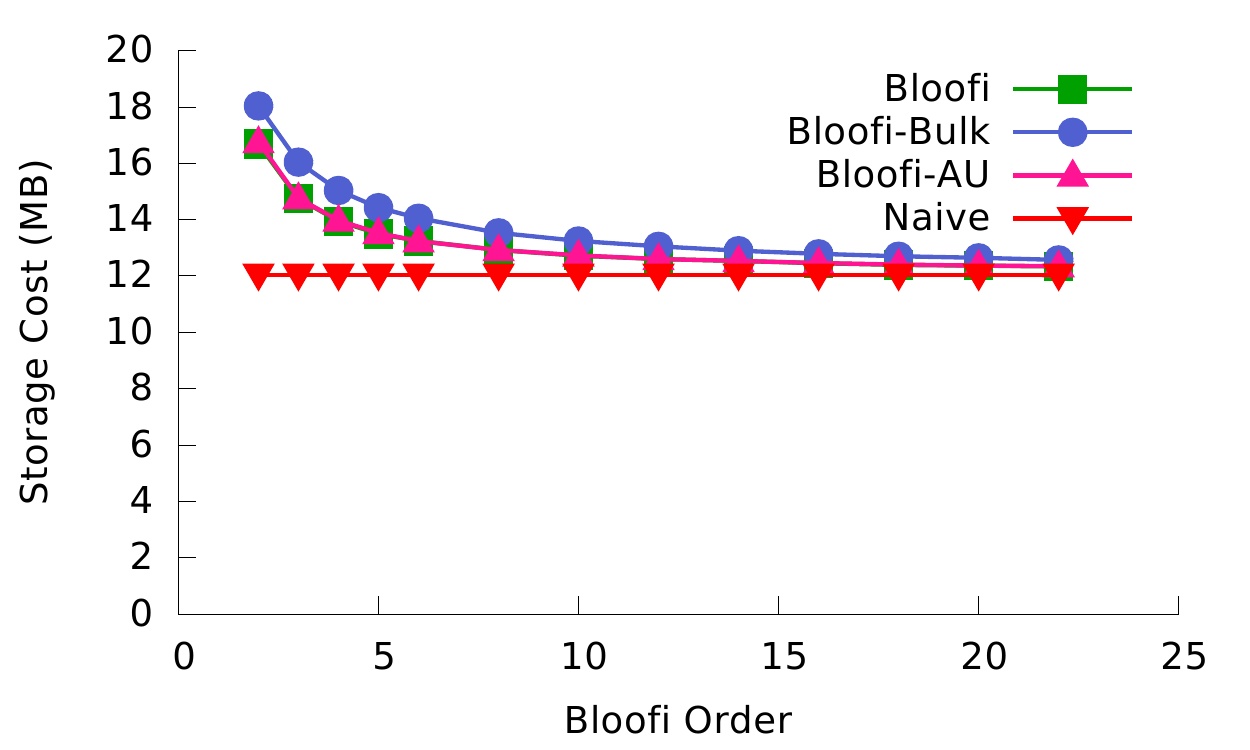}{Storage Cost vs.~$d$ }{fig:spaceCostVsOrder}
{Varying Bloofi Order $d$}

\subsubsection{Varying Bloofi order $d$}

Fig.~\ref{fig:yesSearchTimeVsOrder} and Fig.~\ref{fig:yesSearchesVsOrder} show the search time and search bf-cost as the order of Bloofi increases. The search bf-cost and search time increase as $d$ increases, since the search cost is proportional to $d \log_d N$. For a constant $N$, the function $f(x) = x \log_x N$ is convex, with minimum value when $x = 2.718 (e)$. The best search performance is achieved for Bloofi tree of low order. The search bf-cost obtained from experiments is  close to the ideal search cost for full trees ($2d$ children per non-leaf node), which shows that our algorithms for tree construction perform  well, and not many nodes outside of the path from the root to the answer leaf are checked during search. Bulk construction performs better than incremental construction, due to the global sort. The Bloofi and Bloofi-AU lines are almost identical, which shows that our in-place update maintains the performance of the Bloofi tree.

The storage cost decreases with order (Fig.~\ref{fig:spaceCostVsOrder}), as the number of non-leaf nodes in a Bloofi tree of order $d$ is between $\frac{N-1}{2 d-1}$ and $\frac{N-1}{d-1}$, so higher the order, the lower the overhead of storing the constructed Bloofi tree. While the storage cost is lower for higher orders, the search cost is higher (Fig.~\ref{fig:yesSearchesVsOrder}), so there is a trade-off between search and storage cost. We expect search performance to be more important than storage in practice, and the storage cost of Bloofi is similar with the Naive case even in the worst case (at most twice as much space is needed for Bloofi as it is needed for the Naive case), so we believe that Bloofi trees of low order will be used more in practice. The storage cost for bulk construction is slightly higher than for incremental construction, because always inserting into the right-most leaf leads to skinnier, taller trees, with more nodes.

Fig.~\ref{fig:insDelCostVsOrder} shows the variation of maintenance bf-cost with the order of Bloofi tree. The insert and delete costs ($O(d \log_d N)$) increase with $d$, while the update cost ($O(\log_d N)$) decreases with $d$.

\subsubsection{Varying Bloom filter size}

We vary the Bloom filter size $m$ by varying the number of expected elements in the Bloom filter $\nbElExpected$ from \num{100} to \num{100000}. $m=\left \lceil{k / \ln 2 * \nbElExpected} \right \rceil$.

Fig.~\ref{fig:yesSearchesVsFilterSize} shows the search bf-cost variation with the Bloom filter size in a system with \num{100000}~total elements (1000 filters with 100 elements each). The search bf-cost for Bloofi is always below Naive cost, and it decreases to the ideal cost as the size of the Bloom filters increases and the false positive probability at the high level nodes in Bloofi decreases. In fact, $O(d \log_d N)$ search cost is achieved as soon as $\probf{}<1$ for the root, even if $\probf{}$ is  close to 1 (\num{0.993} in our experiments for size \num{100992}).

The search time (Fig.~\ref{fig:yesSearchTimeVsFilterSize}) also decreases for Bloofi as $\probf{}$ at the higher levels is reduced, while the search time for Naive case is almost constant. For  small Bloom filter sizes, the search time for the Naive case is slightly lower than using Bloofi, likely due to memory locality. The search time for Flat-Bloofi is lower than for Bloofi when the pruning capabilities of Bloofi are reduced, due to the efficient bit level parallelism exploited by Flat-Bloofi. As the size of the Bloom filters increases, the memory locality of Flat-Bloofi is reduced, and the trend for the search time is upward, while the trend is downward for Bloofi.

\threefiguresnew
{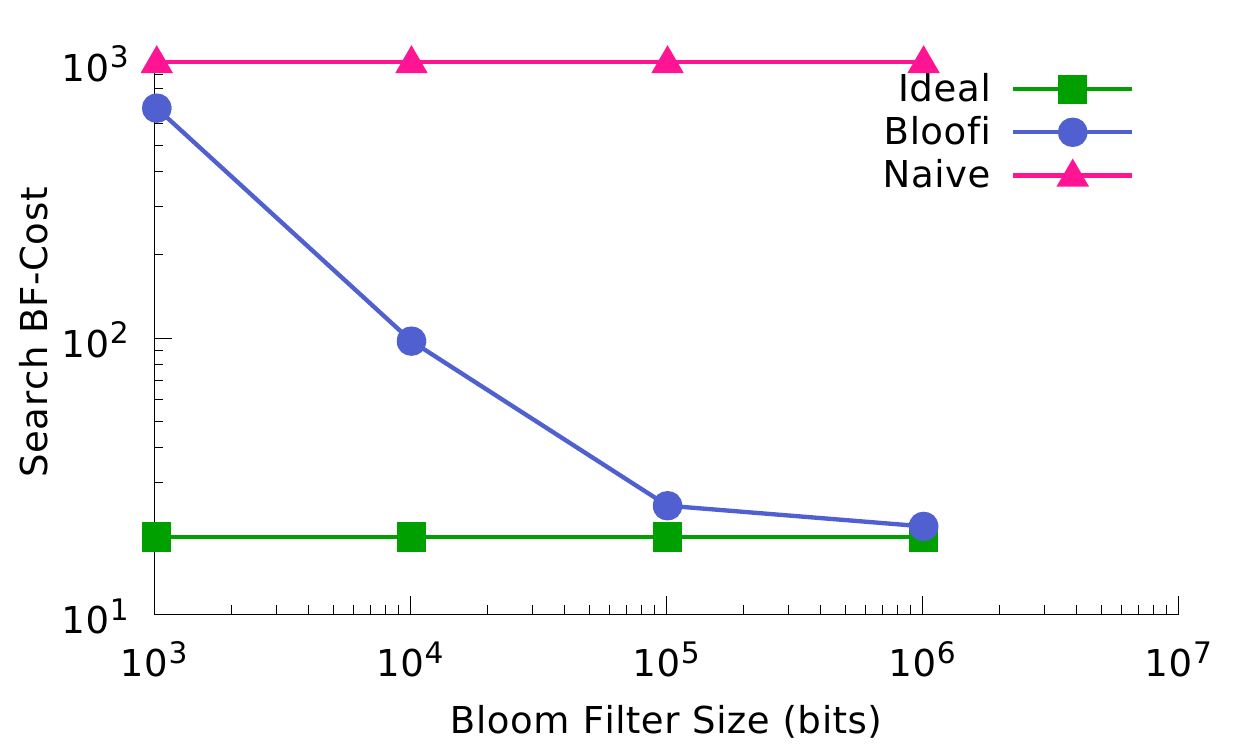}{Search BF-Cost vs.~Filter Size}{fig:yesSearchesVsFilterSize}
{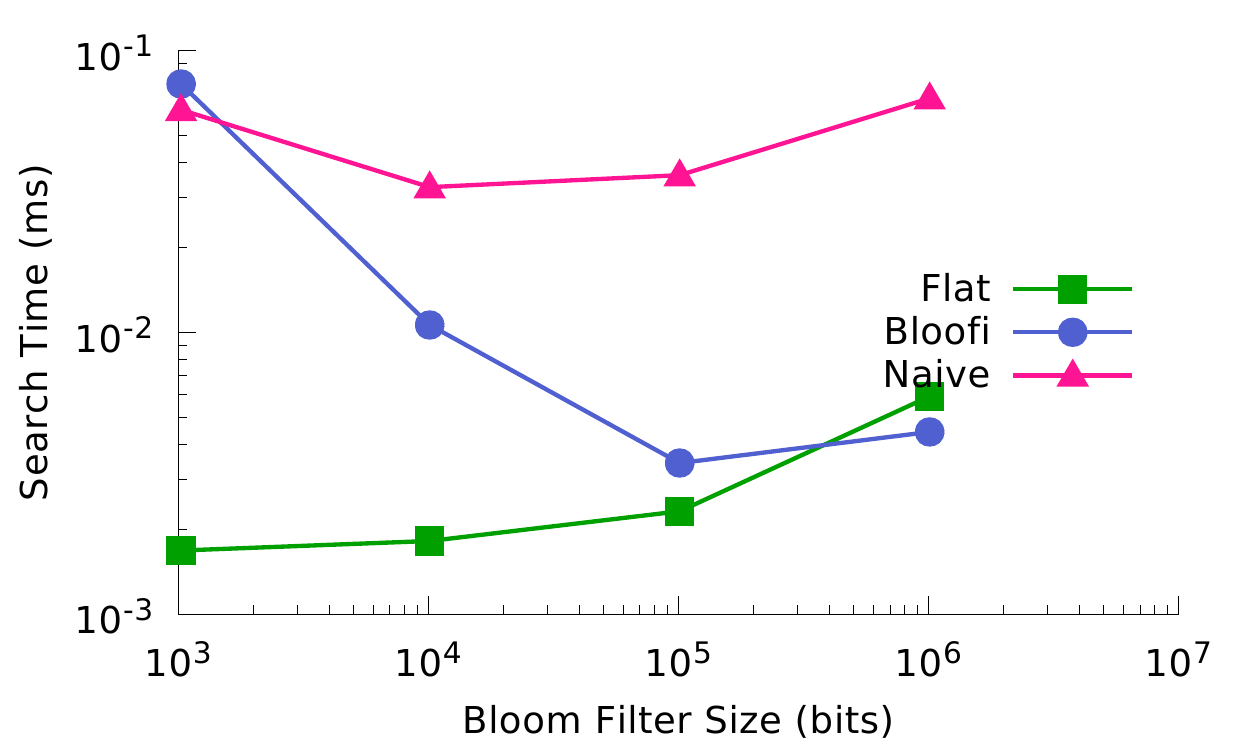}{Search Time vs.~Filter Size}{fig:yesSearchTimeVsFilterSize}
{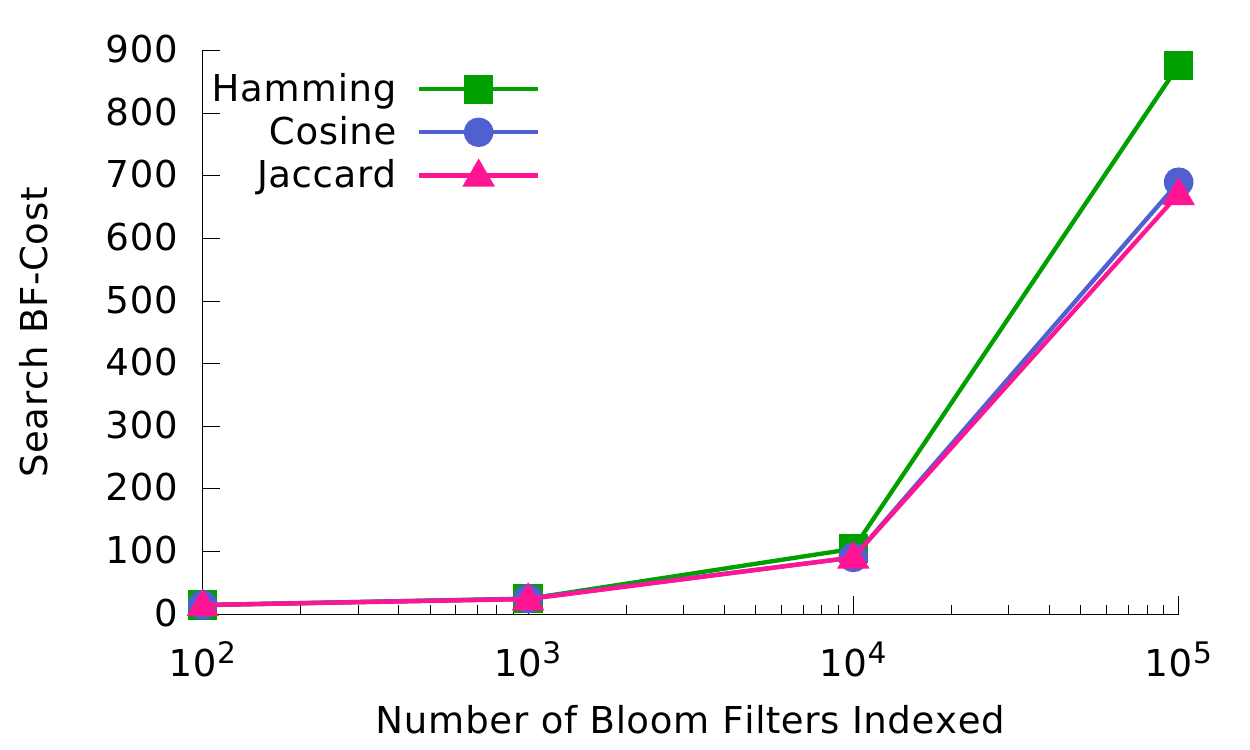}{Search BF-Cost vs.~Metric vs.~$N$}{fig:yesSearchesVsMetricNbBFs}
{}

\subsubsection{Varying the probability of false positives}\label{sec:varyPfalse}

Fig.~\ref{fig:yesSearchesVsPfalse} shows how the search bf-cost for Bloofi increases as we increase the desired probability of false positives in the Bloom filters indexed.  As expected, the bf-cost slowly increases as $\setprobf{}$ and implicitly $\probf{}$ increases, as the pruning capabilities of Bloofi are reduced when $\probf{}$ increases. $\probf{}<1$ at the root for all the $\setprobf{}$ values tested in the experiment, so Bloofi's performance is close to the ideal case.

A more surprising result is shown in Fig.~\ref {fig:yesSearchTimeVsPfalse}. The search time for the Naive case, Bloofi, and Flat-Bloofi is decreasing in the experiment, even if the search bf-cost is increasing. The main reason is the memory locality. As the $\setprobf$ increases, the size of the Bloom filters decreases, leading to better memory locality properties. Flat-Bloofi, with its efficient exploitation of the bit-level parallelism benefits the most from the low memory footprint. This is consistent with the results shown in Fig.~\ref{fig:yesSearchTimeVsFilterSize}. A second reason for the decrease in the search time for all cases is that the number $k$ of hash functions used for the Bloom filters decreases as $\setprobf$ increases. $k = \left \lceil{- \ln \setprobf / \ln 2} \right \rceil$ and decreases from 24 to 4 in the experiment. As fewer hash functions are used, fewer bits in the Bloom filters need to be checked for matches. This contributes to the faster search time.

\subsubsection{Varying the number of elements $n$}\label{sec:varyn}

Fig.~\ref{fig:yesSearchTimeVsNbEl} shows how the search time varies with the increase in $n$, the number of elements in the indexed Bloom filters, as the size of the Bloom filters remains constant. Flat-Bloofi provides the best performance in this case, as it benefits from the memory locality and bit-level parallelism, and the performance is almost constant regardless of the actual number of elements in the Bloom filter. Search time for Bloofi increases only slightly as long as $\probf < 1$ at the root, even when the number of elements in the system is ten times larger than the expected number of elements in a Bloom filter (for $n$ = 100, there are \num{100000} elements in the system, and $nbElExp$ = \num{10000}). Bloofi's performance degrades as its pruning capabilities are reduced when $\probf=1$ at the root. The search bf-cost for Bloofi varies in a similar way (not shown).

\threefiguresnew
{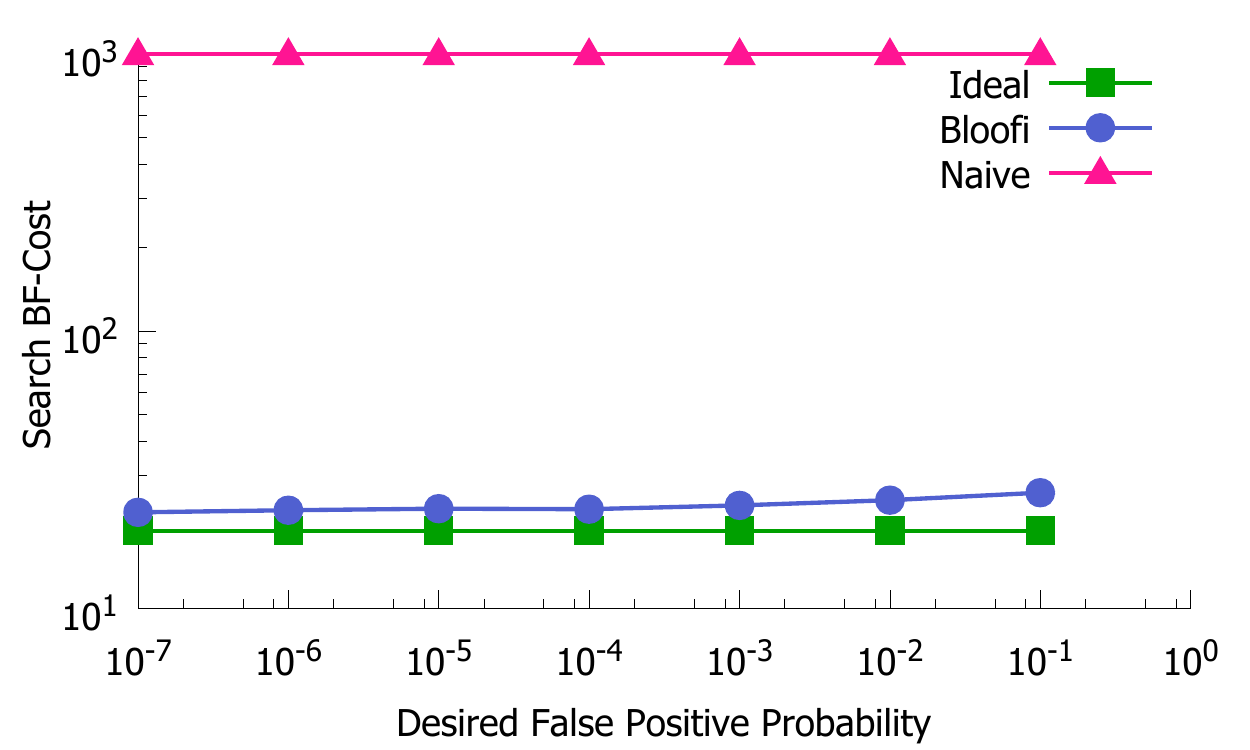}{Search BF-Cost vs.~$\setprobf$}{fig:yesSearchesVsPfalse}
{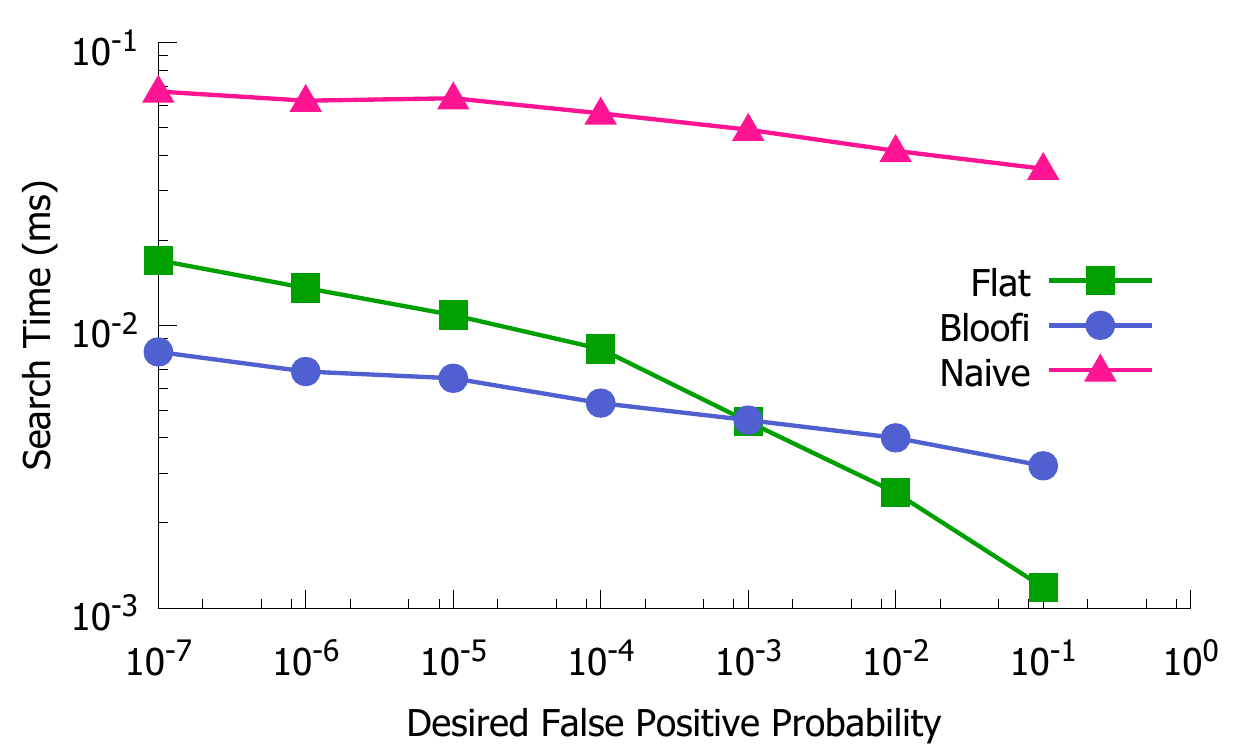}{Search Time vs.~$\setprobf$}{fig:yesSearchTimeVsPfalse}
{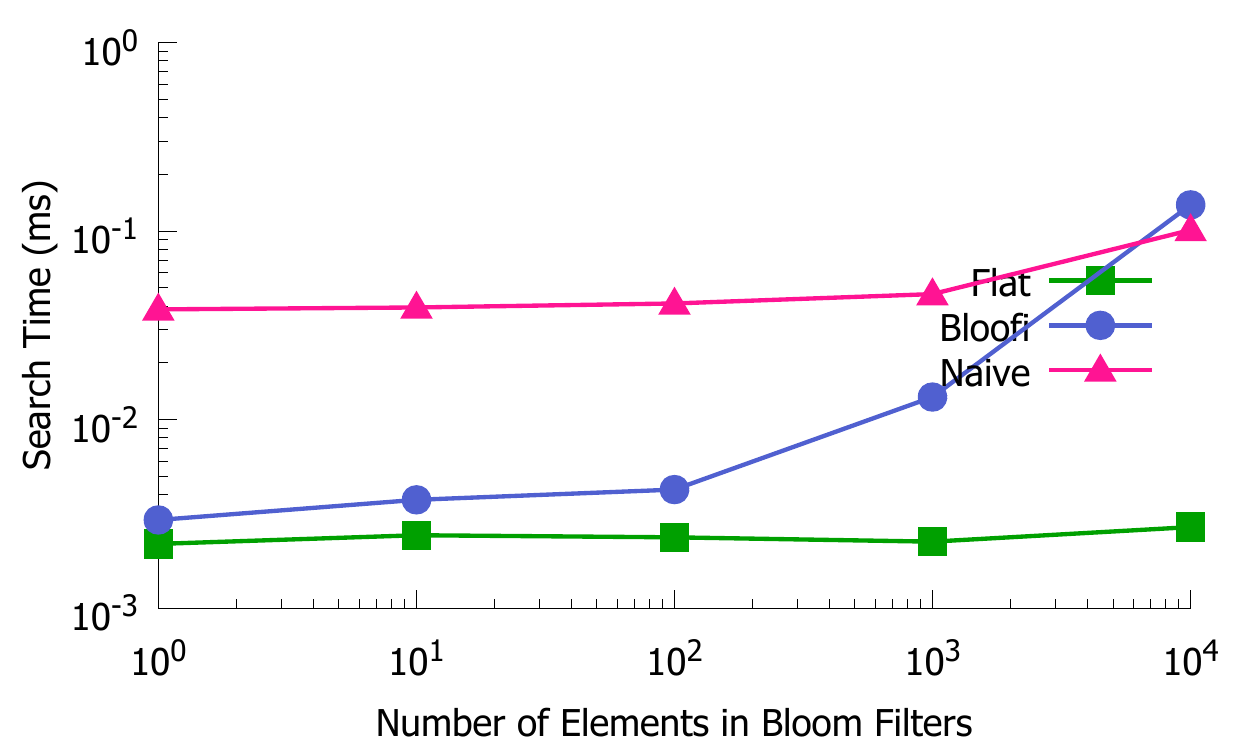}{Search Time vs.~$n$}{fig:yesSearchTimeVsNbEl}
{}

\subsubsection{Varying similarity metric}\label{sec:varyMetric}
During Bloofi construction, we use a similarity metric to measure the similarity between two Bloom filters, and this metric determines the location of the Bloom filters in Bloofi. In this section we experiment with several metrics:
\begin{itemize}
\item $\text{Ham\-ming}(A,B) = |A\ \mathrm{xor}\ B|$,
\item $\text{Jac\-card}(A,B) = 1-|A\ \mathrm{and}\ B|/|A\ \mathrm{or}$ $\ B|$,
\item and $\text{Cosine}(A,B) = 1 - |A\ \mathrm{and}\ B|/\|A\|_2 \times \|B\|_2$,
\end{itemize}
where $|A|$ is the number of 1s in the bit array $A$. The search bf-cost is similar for all the metrics (Fig.~\ref{fig:yesSearchesVsMetricNbBFs}). When Jaccard is used, the search cost is a little lower, but the differences are small, and they might be just due to chance. A similar trend was obtained for search time, as shown in Fig.~\ref{fig:yesSearchTimeVsMetricNbBFs}. The search time is the lowest for Jaccard, and highest for Hamming.

\threefiguresnew
 {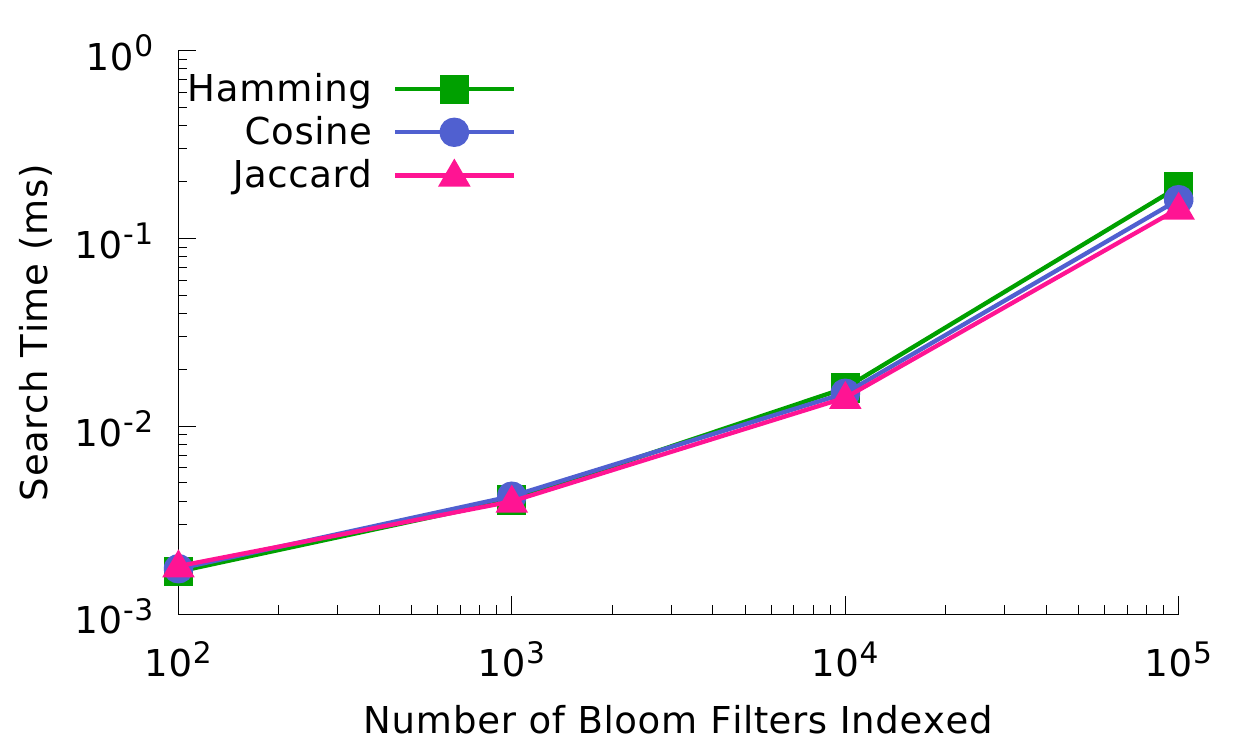}{Search Time vs.~Metric vs.~$N$}{fig:yesSearchTimeVsMetricNbBFs}
 {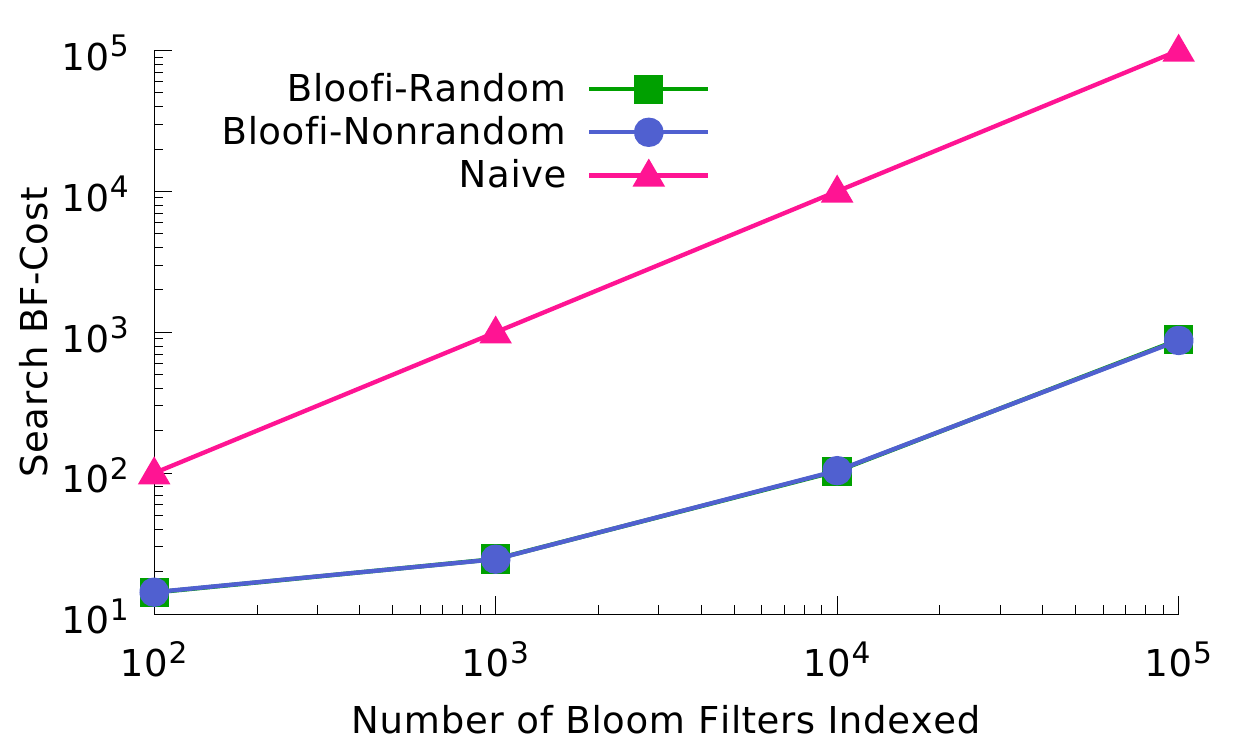}{Search BF-Cost vs.~Data Distribution vs.~$N$}{fig:yesSearchesVsDistrNbBFs}
{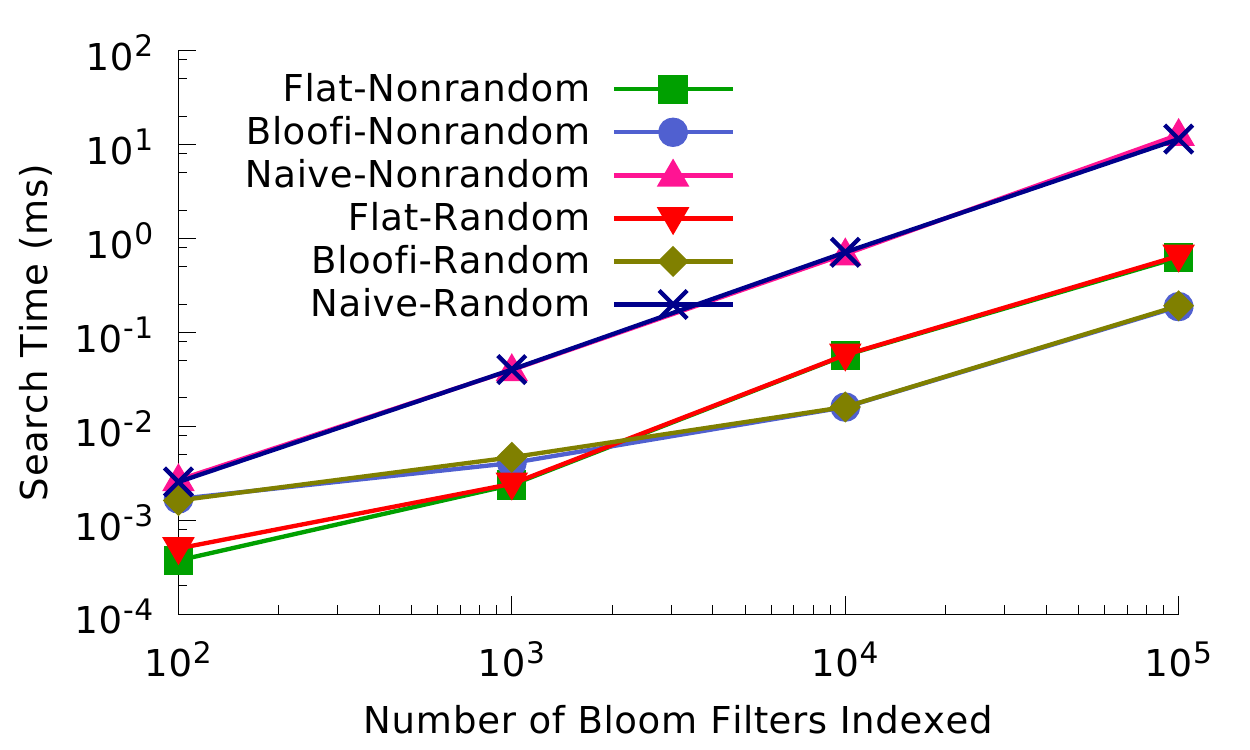}{Search Time vs.~Data Distribution vs.~$N$}{fig:yesSearchTimeVsDistrNbBFs}
{}

\subsubsection{Varying underlying data distribution}
 We use the following distributions for the underlying data in the Bloom filters: $nonran\-dom$, in which we insert in each Bloom filter $0 \le i \le N$ the $\nbEl{}$ integers in the range $[i\times \nbEl{}, i\times \nbEl{}+\nbEl{})$, and $random$, in which we insert in each Bloom filter $\nbEl{}$ integers randomly chosen from a random range assigned to that Bloom filter. $\nbEl{}$ is the actual number of elements in the Bloom filters, which is set to 100 for these experiments. In the $nonrandom$ case, there is no overlap between the sets represented by different Bloom filters, while in the $random$ case, we could have overlap between the sets. We expect that in a real scenario, the $random$ case is more common. The search cost is shown in Fig.~\ref{fig:yesSearchesVsDistrNbBFs} and the search time is shown in Fig.~\ref{fig:yesSearchTimeVsDistrNbBFs}. We expected the performance in the $random$ case to be a little worse than in the $nonrandom$ case, since multiple Bloom filters might match a queried object. However, the overlap of ranges in the $nonrandom$ case was not enough to lead to a substantial increase in the search cost or search time for the $nonrandom$ case, even if the number of results was indeed larger for $nonrandom$ case. The search performance is  similar for both distributions used, which shows that Bloofi is robust to the underlying data distribution.

\subsection{Experimental results conclusion}
Our experimental evaluation shows that:
\begin{itemize}\itemsep1pt \parskip0pt \parsep0pt
\item For Bloofi, the search cost increases logarithmically with $N$, but gets worse if the false positive probability at the root is 1. For optimal performance, the size of the Bloom filters should be based on the estimate for the total number of elements in the entire system. If the false positive probability at the root is below 1, even if  close to 1, the search cost is $O(d \log_d N)$.
\item For Flat-Bloofi, the search time increases with $N$. For low number of Bloom filters, search time for Flat-Bloofi is lower than for Bloofi.
\item Search costs for Bloofi increase with order, since search cost is in $O(d \log_d N)$, so low order is preferred for Bloofi.
\item Storage cost for Bloofi is $O(N+N/d)$
\item For Bloofi, insert cost and delete cost are $O(d \log_d N)$, and update cost is $O(log_d N)$
\item Bulk construction gives slightly better search results, since the trees constructed with bulk construction are skinnier. However, differences in performance between bulk and incremental construction are small. The cost of sorting, which is the first step in bulk construction, is $O(N^2)$ in our implementation, so the operation is expensive. Our experimental results show that the Bloofi algorithm for insertion produces a tree  close to the tree obtained by using a global ordering of Bloom filters.
\item The distance metric used to compare the ``closeness'' between Bloom filters does not have a big effect on the search performance of the resulting Bloofi tree, but the Jaccard distance seems to lead to best search performance.
\end{itemize}

\section{Related work}\label{sec:relatedWork}

Part of the existing work related to Bloom filters~\cite{Cohen2003spectralBloomFilter, Deng2006stableBloomFilters, Dutta2012inteligentCompression, fan98webCaching, Mitzenmacher2001compressedBloomFIlter} is concerned with extending or improving the Bloom filters themselves and does not deal with the problem of searching through a large set of Bloom filters.

Applications of the Bloom filters to web caching~\cite{fan98webCaching} use multiple Bloom filters, but their number is in general small and a linear search through the Bloom filters is performed. \cite{Gong05bloomfilter-based} introduces an XML filtering system based on Bloom filters, and uses a data structure similar to Flat-Bloofi. Mullin~\cite{mullin90tose} uses Bloom filters to reduce the cost of semijoins in distributed databases. Most applications of Bloom filter~\cite{Broder02networkapplications,Chang:2008:BDS:1365815.1365816}, use the Bloom filters directly and do not search through a large number of Bloom filters.

B+~trees~\cite{Comer:1979:UBT} inspired the implementation of the Bloofi tree. However, each node in Bloofi has only one value, and the children of a node do not represent completely disjoined sets, so multiple paths might be followed during search.

A closely related idea to that of the multidimensional Bloom filter problem dates back to superimposed codes~\cite{1455812} and descriptor files~\cite{Pfaltz:1980:PRU:359007.359013}: descriptors (effectively bit arrays) are searched through large files
with a hierarchical data structure constructed by OR-ing the descriptions. When applied to set-valued attributes, these
descriptors are called
``signatures'': each possible value is mapped to a
fixed number of bit positions that are set to true
if the value is present, sets are constructed by
setting all the bit positions corresponding to all values.
The
S-tree~\cite{Deppisch:1986:SDB:253168.253189,390248}
is an implementation of this idea which resembles Bloofi: signatures
are organized in a similar tree structure using the Hamming distance
to aggregate similar signatures.
S-trees are primarily used to index set-valued attributes.
Our problem differs in that we are provided Bloom filters in a distributed setting, and must quickly locate which Bloom filter matches a given query.

Our Flat-Bloofi approach is similar to the
word-parallel, bit-serial (WPBS) approach used for scanning signatures~\cite{Ahuja:1980:APP:800053.801929,46285}, except that our data structure is in-memory, supports moderately fast deletion and does not require dedicated hardware. It is also similar to bitsliced signature files~\cite{Sacks-Davis:1987:MAM:32204.32222,Zobel:1998:IFV:296854.277632} as they apply to text indexing.

\section{Conclusions and future work}\label{sec:conclusions}
 We introduced Bloofi, a hierarchical index structure for Bloom filters. By taking advantage of intrinsic properties of Bloom filters, Bloofi reduces the search cost of membership queries over thousands of Bloom filters and efficiently supports updates of the existing Bloom filters as well as insertion and deletion of filters. Our experimental results show that Bloofi scales to tens of thousands of Bloom filters, with low storage and maintenance cost. In the extreme worst case, Bloofi's performance is similar with not using any index ($O(N)$ search cost), and in the vast majority of scenarios, Bloofi delivers close to logarithmic performance even if the false positive probability at the root is close to (but less than) one.
 When there are fewer Bloom filters, we found that an alternative
 designed to exploit bit-level parallelism (Flat-Bloofi) fared better.

We could pursue more advanced applications. For example, Bloofi and Flat-Bloofi could be applied when only current data from a moving window needs to be maintained. If the window of interest is multiple days, separate Bloom filters can be constructed for data collected each day. Old Bloom filters can be deleted from Bloofi, and new filters inserted, so only the objects in the current window are represented in the Bloofi index.

\section{Acknowledgements}
Special thanks to C.~Crainiceanu for insightful discussions while developing Bloofi, and D.~Rapp, A.~Skene, and B.~Cooper for providing motivation and applications for this work.

\section{References}

\balance
\bibliographystyle{abbrv}
\bibliography{Bloofi}
\end{document}